\documentclass[sigplan,nonacm]{acmart}

\AtBeginDocument{%
  }

\usepackage{multirow}
\usepackage{algpseudocode}
\usepackage{algorithm}

\setcopyright{none}
\begin{document}

\title{A Readiness-Driven Runtime for Pipeline-Parallel Training under Runtime Variability}

\author{
Ruitao Liu$^{1}$ \quad Xinyang Tian$^{1}$ \quad Shuo Chen$^{1}$ \quad Tingrui Zhang$^{1}$\\
Guang Yang$^{1}$ \quad Alan Zhao$^{2}$ \quad Wei Xu$^{1}$\\
$^{1}$Tsinghua University \quad $^{2}$Scitix AI\\
{\footnotesize\texttt{\{liurt23,tianxy22,s-chen25,zhang-tr22,yangg22\}@mails.tsinghua.edu.cn}}\\[-0.25em]
{\footnotesize\texttt{Alan02@scitix.ai, weixu@tsinghua.edu.cn}}
}
\renewcommand{\shortauthors}{Liu et al.}
\renewcommand\footnotetextcopyrightpermission[1]{}

\begin{abstract}
    Pipeline parallelism is a key technique for scaling large-model training, but modern workloads exhibit runtime variability in computation and communication. Existing pipeline systems typically consume static, profiled, or adaptively generated schedules as pre-committed execution orders. When realized task readiness diverges from the pre-committed order, stages may wait for not-yet-ready work even though other executable work is available, creating stage misalignment, idle bubbles, and reduced utilization.

    We present Runtime-Readiness-First Pipeline (RRFP), a readiness-driven runtime for pipeline-parallel training. RRFP changes how schedules are consumed at runtime: instead of treating a schedule as a sequence that stages must wait to follow, it treats the schedule as a non-binding hint order for ranking currently ready work. To support this model, RRFP combines message-driven asynchronous communication, lightweight tensor-parallel coordination for collective consistency, and ready-set arbitration for low-overhead dispatch.

    We implement RRFP in a Megatron-based training framework and evaluate it on language-only and multimodal workloads at up to 128 GPUs. RRFP improves over fixed-order pipeline baselines across all settings. Using the BFW hint, RRFP achieves up to 1.77$\times$ speedup on language-only workloads and up to 2.77$\times$ on multimodal workloads. In cross-framework comparisons, RRFP with the default BF hint outperforms the faster available external system by up to 1.84$\times$ while preserving training correctness.
\end{abstract}


\begin{teaserfigure}
    \includegraphics[width=\textwidth]{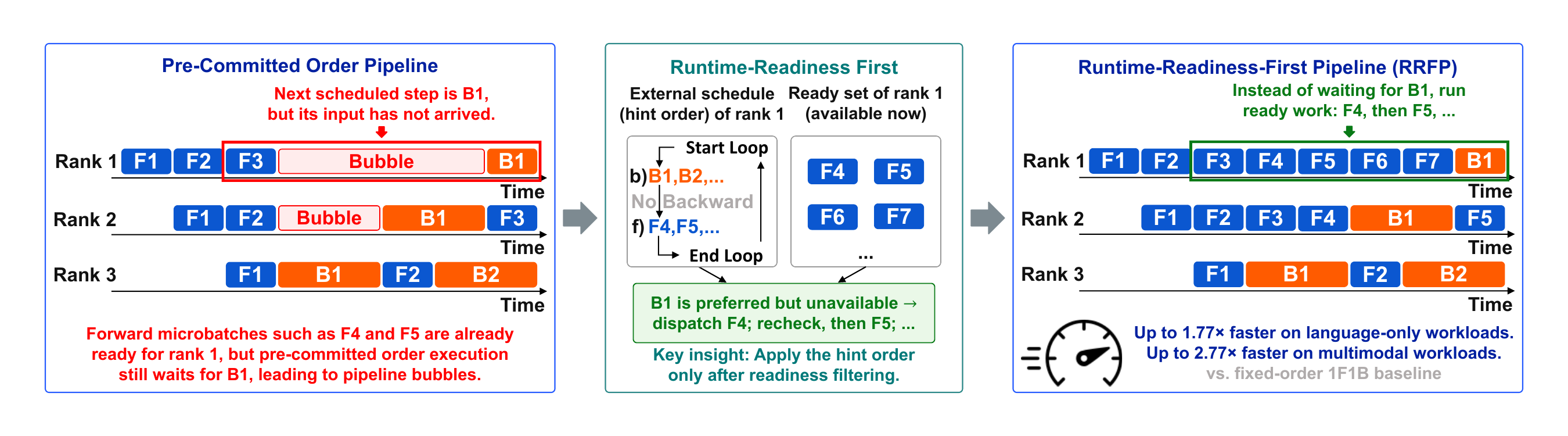}
    \caption{\textbf{RRFP overview.}
        RRFP treats a pre-committed pipeline schedule as a non-binding hint order over currently ready work rather than as an execution sequence to wait for. By skipping unavailable tasks and dispatching ready work, RRFP reduces bubbles and stage misalignment, improving end-to-end training performance over fixed-order 1F1B by up to 1.77$\times$ on language-only workloads and 2.77$\times$ on multimodal workloads when using the BFW hint.}
    \label{fig:teaser}
    \Description{RRFP.}
\end{teaserfigure}

\maketitle
\pagestyle{plain}

\section{Introduction}

Large-scale distributed training has become the standard approach for training modern deep learning models~\cite{shoeybi2020megatronlmtrainingmultibillionparameter,narayanan2021efficientlargescalelanguagemodel,aminabadi2022deepspeedinferenceenablingefficient,sergeev2018horovodfasteasydistributed,wang2025watosefficientllmtraining,lian2025superoffloadunleashingpowerlargescale,295595,Li2021,Lin2026}, especially large language and multimodal models~\cite{brown2020languagemodelsfewshotlearners,grattafiori2024llama3herdmodels,bai2025qwen25vltechnicalreport,bai2025qwen3vltechnicalreport,chen2025parallelscalinglawlanguage,alayrac2022flamingovisuallanguagemodel} that exceed the memory and compute capacity of a single device, making distributed execution necessary. Among various parallelization strategies, \textbf{pipeline parallelism} is a key technique for scaling such models across multiple accelerators~\cite{huang2019gpipeefficienttraininggiant,narayanan2021memoryefficientpipelineparalleldnntraining,10.1145/3669940.3707220,lamypoirier2023breadthfirstpipelineparallelism,pentyala2024paftparalleltrainingparadigm,10792656}.

However, modern distributed training workloads exhibit runtime variability in both computation and communication. Factors such as computation and communication jitter introduce temporal variations that are difficult to predict, especially at scale~\cite{zhao2019dynamicstalesynchronousparallel,Tyagi_2020}. In addition, model- and input-dependent execution can further amplify variation across microbatches, particularly in multimodal pipelines~\cite{Xue_2026,feng2025optimusacceleratinglargescalemultimodal,wang2025spindleefficientdistributedtraining}. As a result, latency across pipeline stages and overall system performance evolve over time during execution.

This dynamic behavior exposes a fundamental limitation in how pipeline execution is currently structured. Modern systems typically organize pipeline execution around an order chosen before the corresponding work is executed. Prior work has explored both offline profiling~\cite{wang2025spindleefficientdistributedtraining,feng2025optimusacceleratinglargescalemultimodal,lin2024nnscaler,wang2025semanticawareschedulinggpuclusters} and online adaptation~\cite{xue2025pipeweaveraddressingdatadynamicity,Xue_2026,jiang2025flexibleprogrammablepipelineparallelism,guo2025adaptisreducingpipelinebubbles,10.1145/3492321.3519563,10.1145/3635035.3635045,yang2020pipemareasynchronouspipelineparallel} to improve pipeline execution order. These techniques provide valuable scheduling signals and can effectively reduce bubbles when their assumptions match realized execution. However, existing systems still ultimately produce or select an execution order that the runtime treats as a \emph{pre-committed execution order}: stages are expected to follow the chosen sequence during execution, even when readiness has changed.

Under runtime variability, this pre-commitment becomes fragile. Task readiness can change at finer granularity than scheduling decisions, so the pre-committed order may become stale and inconsistent with realized execution. When this happens, stages may wait for not-yet-ready work even though other executable work is available, creating stage misalignment, idle bubbles, and reduced utilization. Thus, the limitation does not arise from any particular scheduling strategy, but from making a planned order a requirement for runtime progress. Even a short mismatch between the planned order and the current ready set can force a stage to idle despite available executable work.

To decouple runtime progress from a fixed order chosen ahead of execution, we propose a \emph{runtime-readiness-first} design. At each stage, the runtime first constructs the current ready set and uses a schedule only as a non-binding hint order to rank ready candidates. If a higher-ranked task is not ready, the runtime skips it and dispatches another ready task instead of waiting. The key change is not to compute a different schedule, but to change how schedules are consumed at runtime: schedule guidance is preserved without enforcing a planned fixed execution order, broadening the space of schedules usable as runtime hints.

We instantiate this design with Runtime-Readiness-First Pipeline (\textbf{RRFP}), a runtime system for correct out-of-order pipeline execution based on task readiness. RRFP combines three mechanisms: message-driven asynchronous communication decouples data transfer from the compute loop and updates ready buffers as messages arrive; lightweight tensor-parallel coordination preserves collective consistency when tensor-parallel ranks observe different ready sets; and ready-set arbitration scans the non-binding hint order over ready work and dispatches the selected task with low overhead. Together, these mechanisms preserve the benefits of schedule guidance when it matches realized readiness, while falling back to readiness-driven progress without blocking on stale choices.

RRFP extends Megatron-LM~\cite{narayanan2021efficientlargescalelanguagemodel} as a configurable runtime layer, without modifying model definitions or training workflows. Our evaluation focuses on end-to-end performance, robustness to jitter, and runtime breakdown under dynamic execution conditions. This design lets existing schedules remain useful when their predicted order matches realized readiness, while avoiding stalls when it does not. It therefore separates schedule quality from runtime progress.

\paragraph{Contributions.}
Overall, this paper makes the following contributions:
\begin{itemize}
    \setlength{\leftskip}{-0.5em}
    \item We identify a mismatch between pipeline execution orders chosen ahead of time and realized task readiness under computation and communication variability.
    \item We introduce RRFP, a readiness-driven runtime that decouples pipeline progress from fixed global execution sequences while consuming existing schedules as non-binding hint orders over the ready set.
    \item We develop three mechanisms for correct, low-overhead out-of-order execution: message-driven communication decouples data transfer from computation, tensor-parallel coordination preserves collective consistency, and ready-set arbitration dispatches ready work.
    \item We show that RRFP improves end-to-end training performance. Using the BFW hint, RRFP achieves up to 1.77$\times$ speedup on language-only workloads and up to 2.77$\times$ on multimodal workloads over the fixed-order 1F1B. In representative cross-framework comparisons, RRFP's default BF configuration outperforms the faster applicable external baseline by up to 1.84$\times$, while preserving training correctness.
\end{itemize}

\section{Characterizing Runtime Variability}

\subsection{Runtime Variability Beyond Workload Dynamicity}

Execution time in distributed training is influenced by multiple sources of variability. A substantial body of prior work emphasizes \emph{workload dynamicity}~\cite{xue2025pipeweaveraddressingdatadynamicity,Xue_2026,DBLP:conf/middleware/RochaMCFBS20,sun2024seq1f1befficientsequencelevelpipeline}, where input composition (e.g., sequence length or modality mix) changes per microbatch and leads to execution-time differences.

In this paper, we focus on a complementary source that appears even when the workload is fixed: \emph{runtime variability}, including kernel-level jitter, resource contention, and communication fluctuation~\cite{zhao2019dynamicstalesynchronousparallel,Tyagi_2020}. These effects introduce timing variation in both computation and communication under identical model and input settings.

We show that this source of unpredictability challenges existing fixed-order approaches and is sufficient to disrupt pipeline execution, even under fixed configurations.

Specifically, runtime variability leads to run-to-run differences in computation and communication latency, which in turn changes when tasks become ready and shifts favorable execution decisions. We characterize these effects below.

\begin{figure}[!tbp]
    \centering
    \includegraphics[width=\linewidth]{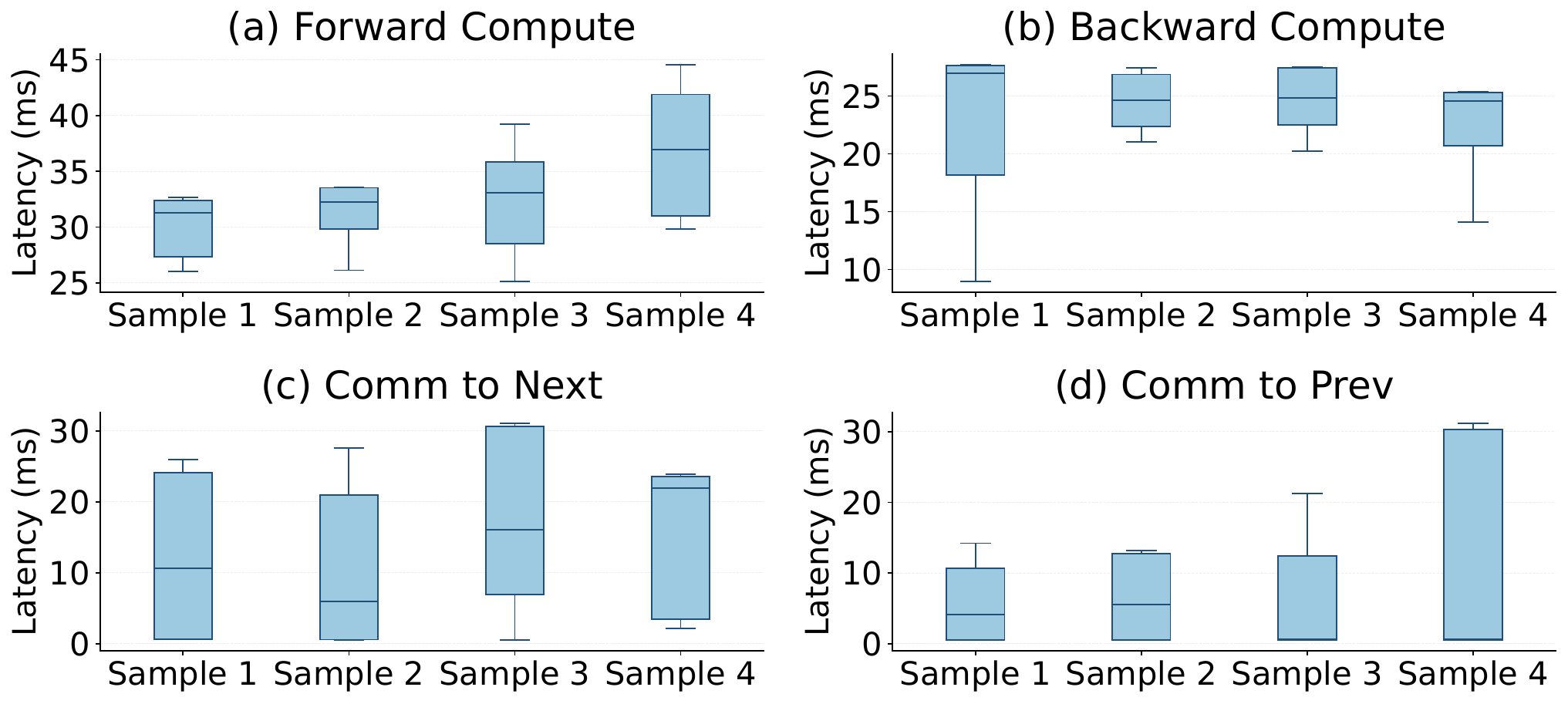}
    \caption{\textbf{Run-to-run latency variability under identical training conditions.} We repeat the same training procedure 300 times with the same model, input data, data order, hardware, pipeline configuration, and number of training iterations. We select 4 fixed samples and track them across all runs. Each boxplot summarizes the latency distribution of one tracked sample for one event type. Boxes span the interquartile range, from p25 to p75, with the center line indicating the median. Whiskers indicate p5 and p95 latencies.}
    \label{fig:1}
    \Description{Boxplots of latency over 300 repeated runs for four fixed samples across forward computation, backward computation, and inter-stage communication. Boxes span p25 to p75 with median lines, and whiskers indicate p5 and p95.}
\end{figure}

\begin{figure}[!tbp]
    \centering
    \includegraphics[width=\linewidth]{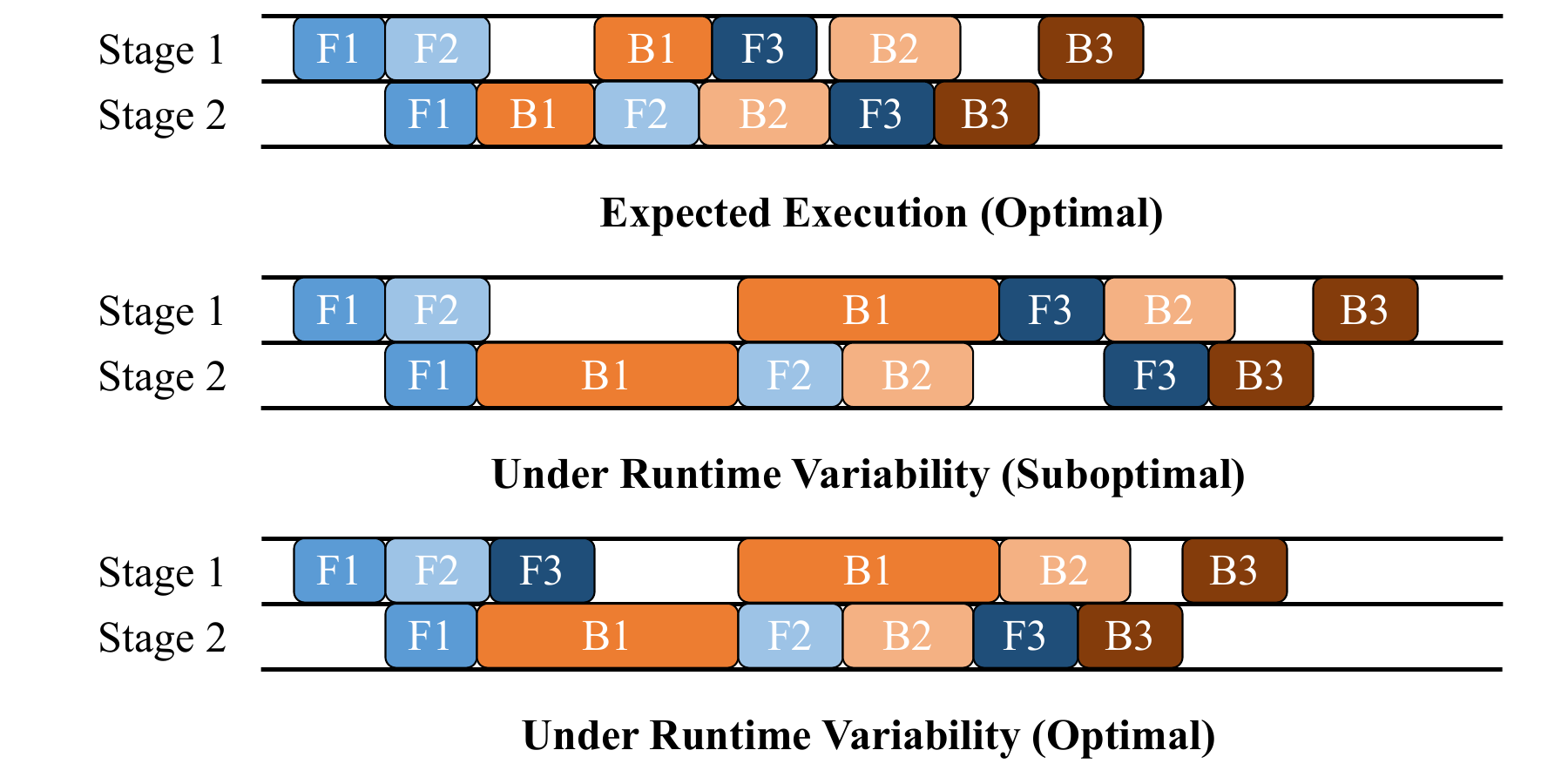}
    \caption{\textbf{Effect of runtime variability on execution decisions.}
        A planned order can be favorable under expected timing but become suboptimal when realized timing changes. Here, scheduling B1 before F3 creates idle gaps under runtime variability, while dispatching the ready F3 task earlier better matches the realized state and reduces idle time.}
    \label{fig:2}
    \Description{Illustration showing that runtime variability can make a planned execution order suboptimal, while another ready-task order reduces idle time.}
\end{figure}

\subsection{Runtime Variability Changes Task Readiness}
\label{sec:2.2}

To make the measurement concrete while isolating runtime variability from input and configuration changes, we repeat one fixed 1F1B training procedure 300 times under identical settings. The repeated runs use PP=8 and TP=1, and keep the model, input data, data order, hardware placement, pipeline configuration, number of training iterations, and 1F1B execution order unchanged. We select four fixed samples and track the same samples across all runs. For each sample in each run, we record four event latencies: forward computation, backward computation, communication to the next pipeline stage, and communication to the previous pipeline stage. Therefore, the measured spread reflects run-to-run runtime effects rather than workload or configuration differences.

Figure~\ref{fig:1} shows that noticeable latency variation remains even under these fixed conditions. The spread appears across multiple samples and across both computation and communication events, showing that jitter is not limited to a particular sample or operation type. Quantitatively, the normalized p95--p5 spread, $(p95-p5)/p50$, reaches \textbf{0.73} for computation and \textbf{58.74} for communication, while the normalized interquartile spread, $(p75-p25)/p50$, reaches \textbf{0.35} for computation and \textbf{50.80} for communication.

These results show that runtime variability persists even when the input data and execution context are fixed. This makes it fragile to treat profiled or estimated execution times as a hard basis for runtime ordering, motivating a runtime layer that can react to realized readiness.

This variability matters because pipeline decisions depend critically on when tasks become ready. Figure~\ref{fig:2} illustrates that an order favored under planned timing, such as scheduling B1 before F3, can become suboptimal when runtime jitter changes realized timing. In the example, B1 takes longer than expected, so following the same committed order creates idle gaps. Dispatching the ready F3 task earlier instead better matches the realized execution state and reduces idle time. Such local waiting can also delay neighboring stages, amplifying bubbles through pipeline dependencies.

This motivates a runtime-readiness-first design. Instead of waiting for a planned task that is not yet ready, the runtime re-evaluates the ready set as execution state evolves.

\section{Method Overview}

\subsection{Execution Model and Constraints}

We model pipeline execution within a training iteration as a dependency-constrained execution process over computation tasks. For each task $v$, let $s_v$ and $e_v$ denote its start and end time, and let $p_v$ denote its execution time, such that $e_v=s_v+p_v$. We use iteration completion time, or makespan,
\[
    C_{\max}=\max_{v\in V} e_v,
\]
as the primary efficiency objective.

A task can start only after all required predecessors have completed. These include inter-stage dependencies, where forward tasks depend on activations from the previous stage and backward tasks depend on gradients from the next stage, and intra-stage dependencies, where the backward task of a microbatch at a stage depends on its local forward task. We write these precedence constraints as
\[
    s_v \ge e_u,\quad \forall (u\rightarrow v)\in E,
\]
where $E$ denotes the set of dependency edges. Each pipeline stage can execute at most one task at a time, so tasks assigned to the same stage must be serialized.

This model exposes the runtime decision problem faced by each stage. At any point, only a subset of tasks is executable, and this ready set changes as computation finishes and messages arrive. A fixed global order can therefore become inconsistent with the local ready set observed at runtime.

\subsection{RRFP Overview}

RRFP is a readiness-driven runtime for out-of-order pipeline-parallel training. Its core principle is readiness-first execution: each stage makes progress on currently executable work, while correctness and communication are handled by runtime mechanisms rather than a globally fixed step order.

This execution model introduces three runtime challenges: communication-order divergence across neighboring stages, collective-order divergence across tensor-parallel ranks, and frequent arbitration among ready microbatches. RRFP addresses them with three corresponding mechanisms: (1) asynchronous message-based communication decouples data transfer from computation; (2) tensor-parallel coordination preserves collective consistency by coordinating only within each tensor-parallel group before collective-relevant computation; and (3) ready-set arbitration dispatches executable work using simple priorities while allowing external schedules to act as non-binding hint orders. Figure~\ref{fig:3} summarizes how these mechanisms interact at one pipeline stage. Sections~\ref{sec:4.1}--\ref{sec:4.2} describe the first two mechanisms for correct out-of-order execution, and Section~\ref{sec:5} describes the arbitration layer.

\section{Runtime Mechanisms for Correct Out-of-Order Execution}

RRFP enables stages to execute microbatches according to observed readiness rather than a globally fixed order. Section~\ref{sec:4.1} addresses communication-order divergence across neighboring stages, and Section~\ref{sec:4.2} addresses collective-order consistency across tensor-parallel ranks.

\subsection{Message-Driven Asynchronous Communication}
\label{sec:4.1}

\begin{figure*}[!tbp]
    \centering
    \includegraphics[width=0.92\linewidth]{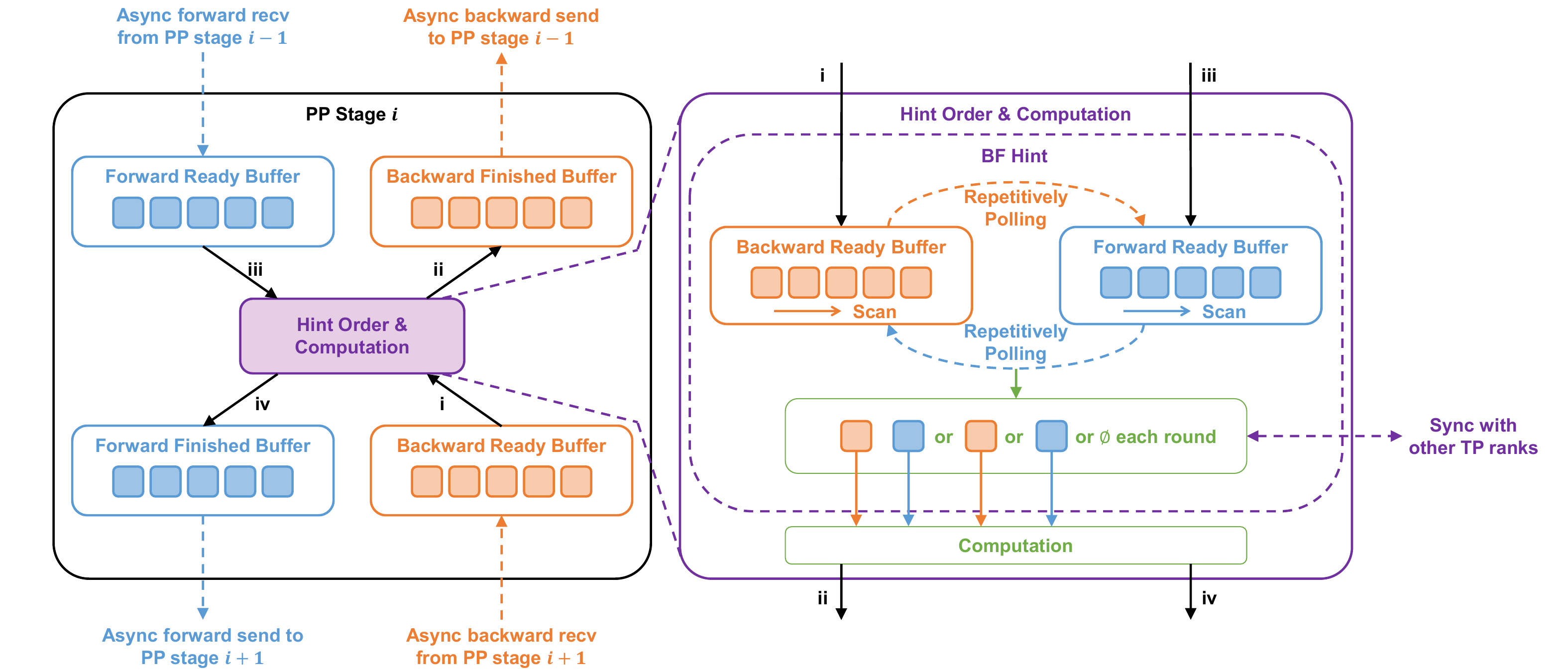}
    \caption{
        \textbf{Workflow of RRFP at one pipeline stage.}
        Asynchronous send/receive threads update forward/backward-ready buffers and drain finished buffers to neighboring pipeline stages.
        The arbitration layer repeatedly polls ready buffers and scans them according to a hint order, shown here with the backward-forward (BF) example, selecting only executable microbatches.
        Tensor-parallel ranks then synchronize the selected work item, and consistent selections proceed to computation.
    }
    \label{fig:3}
    \Description{The workflow of RRFP.}
\end{figure*}

\paragraph{Communication challenge.}
Out-of-order pipeline execution breaks the fixed send/receive order assumed by conventional pipeline communication. In a fixed schedule, neighboring stages issue sends and receives in the same predetermined microbatch order, so communication can be matched by sequence. Under readiness-driven execution, however, stages may progress on different microbatches. A sender may finish and issue transfers for several microbatches before the receiver reaches the corresponding receives, or a receiver may wait for data whose sender is still computing. Low-level communication stacks can move data efficiently, but the runtime must still match each activation or gradient to the correct microbatch and update readiness when messages arrive out of order. If the runtime instead waits for sends and receives in a predetermined sequence, such order divergence can cause stalls or deadlocks.

\paragraph{Message-driven model.}
RRFP decouples data transfer from the compute loop with a message-driven communication model. When a forward or backward task completes, the produced tensor is submitted as a message carrying the microbatch identifier and direction. Receivers use this metadata to route out-of-order arrivals into per-microbatch buffers. A forward task becomes ready when its activation arrives from the previous stage. A backward task becomes ready only after the corresponding local forward task has completed and the gradient from the next stage has arrived. Thus, message arrival directly updates the local ready set.

\paragraph{Runtime implementation.}
RRFP implements the C++ communication backend with dedicated send and receive threads. At each pipeline stage, the runtime consists of a compute thread that executes forward/backward tasks and communication threads that move tensors asynchronously, decoupling computation from data transfer. Each stage maintains four buffers: forward-ready, forward-finished, backward-ready, and backward-finished. The compute thread polls the ready buffers, runs F/B tasks, inserts completed tensors and metadata into the corresponding finished buffers, and wakes sender threads. Sender threads submit outgoing transfers, while receiver threads accept incoming messages, route them by microbatch identifier, and update the ready buffers.

RRFP exposes the buffer-size limit as a configurable parameter and enforces it via a backpressure policy when forward computation runs far ahead of backward computation. The detailed buffer-size policy analysis and deadlock-free guarantee are provided in Appendix~\ref{app:buffer}.

We study the effect of this buffer-size limit on training iteration time in Figure~\ref{fig:buffer} using Qwen3-4B + ViT-Big~\cite{yang2025qwen3technicalreport,radford2021learningtransferablevisualmodels} with global batch size 192. We sweep the buffer limit over $\{4,8,16,32,48\}$. In this configuration, the worst-case demand is 192 entries if forward execution runs far ahead without backpressure. Increasing the limit gives the runtime more room to keep ready work, improving scheduling flexibility and reducing iteration time at small limits. The benefit saturates once the limit reaches 16 entries, far below the worst-case demand, indicating that a moderate limit exposes most readiness-driven opportunities without unnecessary memory growth. We choose 32 as a conservative value within the saturated range. Based on this sensitivity study, all experiments use a buffer-size limit of 32 unless otherwise stated.

\begin{figure}[t]
    \centering
    \includegraphics[width=\columnwidth]{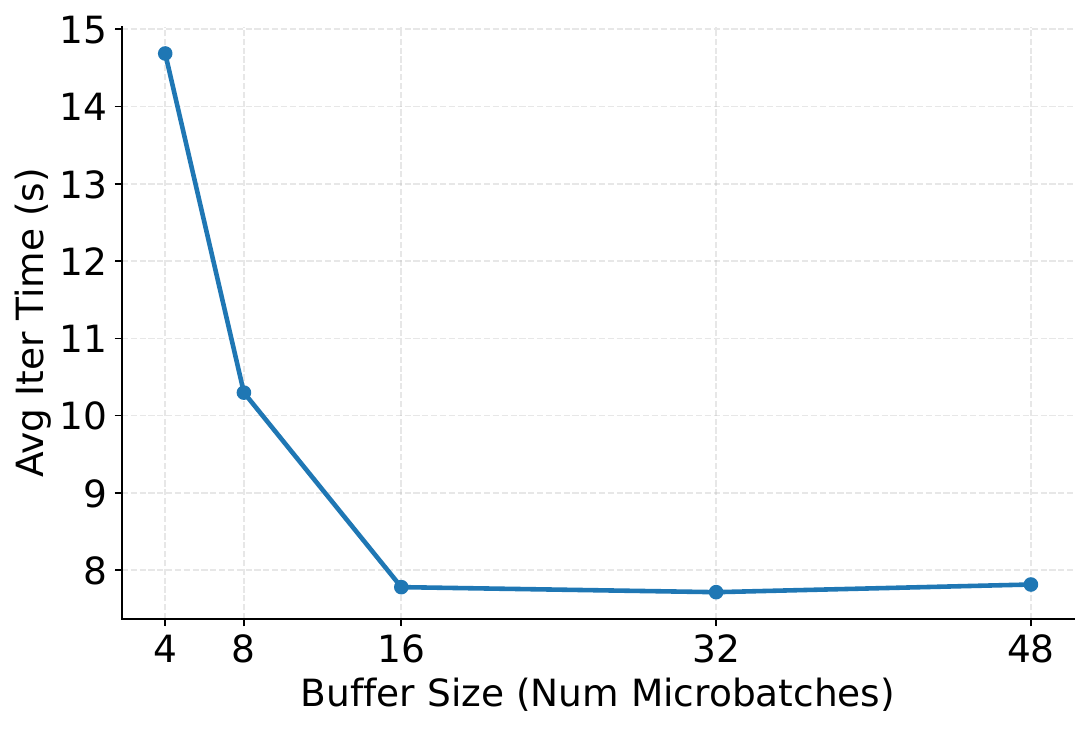}
    \caption{\textbf{Performance sensitivity to buffer-size limit.} With global batch size 192, the worst-case demand is the full batch size if forward execution runs far ahead. Iteration time saturates at a limit of 16 entries, indicating that a moderate limit exposes most readiness-driven opportunities.}
    \label{fig:buffer}
    \Description{Sensitivity of iteration time to the runtime buffer-size limit.}
\end{figure}

\subsection{Tensor-Parallel Coordination}
\label{sec:4.2}

\paragraph{Tensor-parallel constraint.}
Out-of-order execution introduces a correctness constraint under tensor parallelism. In Megatron-LM, ranks in the same tensor-parallel group must invoke collectives in the same microbatch order. Under readiness-driven execution, however, ranks may observe different ready sets and select different microbatches, causing mismatched collective calls, incorrect results, or deadlock.

\paragraph{Coordination protocol.}
RRFP enforces this constraint with a lightweight group-local coordination step. Before a forward or backward task that may invoke tensor-parallel collectives, each rank proposes its selected ready microbatch, and the group exchanges these identifiers using a scalar all-gather. If all ranks agree, execution proceeds. Otherwise, the task is deferred until readiness changes expose a common candidate. This preserves collective-order consistency without forcing a global pipeline order. We provide additional details on the coordination protocol and progress argument in Appendix~\ref{app:tp-progress}.

\section{Ready-Set Arbitration}
\label{sec:5}

This section describes the local arbitration layer that selects the next microbatch to execute from the runtime ready buffer, without placing an expensive scheduler on the critical path. At each pipeline stage, local completion and message-arrival events continuously update forward/backward-ready buffers, while completed tasks enter finished buffers that feed the asynchronous communication runtime in Section~\ref{sec:4.1}. Thus, arbitration operates only over microbatches whose dependencies have actually been satisfied at runtime.

\paragraph{Hint-based arbitration.}
The arbitration layer is parameterized by a hint order $\Pi$ over microbatches and directions. The hint may come from a static schedule, offline profiling, or online adaptation, as long as it defines an ordered preference over candidates. When the compute thread becomes available, RRFP scans $\Pi$ and dispatches the first entry that is present in the current ready buffer. If no candidate is ready, the thread waits for the next readiness event. Completed tasks are inserted into the corresponding finished buffers. In this way, schedules guide local arbitration without requiring the runtime to wait for unavailable work. The loop is lightweight because it scans bounded ready buffers under a fixed hint order rather than solving an online scheduling problem.

\paragraph{Lightweight instantiation.} The arbitration layer is independent of the specific hint order $\Pi$. Different schedules can be used to rank the current ready set. For analysis and evaluation, we use a simple \emph{backward-forward (BF)} hint as the default lightweight instantiation. Inspired by the standard 1F1B pattern, BF first considers a backward-ready microbatch and then a forward-ready microbatch in each arbitration round, subject to runtime readiness. Within each direction, BF follows pipeline dependencies: forward prefers lower-index model chunks, while backward prefers higher-index chunks. This frees model chunks earlier and allows subsequent communication collectives to overlap with other computation work. The detailed BF hint algorithm is presented in Appendix~\ref{app:algo}. Section~\ref{sec:6} characterizes BF under a simplified setting, and Section~\ref{sec:7.7} evaluates the effect of alternative hint orders on training performance.

\section{Analytical Characterization}
\label{sec:6}

We analyze the backward-forward (BF) hint introduced in Section~\ref{sec:5} under a simplified setting. The purpose is not to model the full RRFP runtime, but to understand the behavior of ready-set arbitration in isolation. We therefore consider a non-interleaved pipeline and ignore communication time, tensor-parallel coordination, collectives such as parameter gather, and implementation overhead. This setting lets us relate the local choices made by the BF hint to pipeline progress and stage imbalance. Proofs and detailed bottleneck statistics are provided in Appendix~\ref{app:analysis}.

Let $N$ be the number of stages, $M$ the number of microbatches, and $\mathcal{C}$ the iteration completion time. For each microbatch $j$, let $F_{\max}^j$ and $B_{\max}^j$ denote the slowest forward and backward computation time across all stages, and let $F_{\text{last}}^j$ and $B_{\text{last}}^j$ denote the corresponding computation time at the last stage. We also define two reference makespans: $\mathcal{F}$ is the time needed to finish all forward work when the first stage initially has all forward microbatches available, and $\mathcal{B}$ is the time needed to finish all backward work when the last stage initially has all backward microbatches available. Both forward-only and backward-only reference executions still respect inter-stage pipeline dependencies.

\begin{theorem}[Upper bound for BF hint arbitration]
    \label{thm:6.1}
    The makespan $\mathcal{C}$ of RRFP under the BF hint satisfies
    \[
        \mathcal{C}\le\mathcal{F}+\mathcal{B}
        +\sum_{j=1}^{M-1}\bigl(F_{\max}^j-F_{\text{last}}^j\bigr)
        +\sum_{j=0}^{M-2}\bigl(B_{\max}^j-B_{\text{last}}^j\bigr).
    \]
\end{theorem}
The bound separates two effects: idealized forward and backward pipeline progress and extra delay caused by stage imbalance. To interpret this bound, we compare it with the total amount of work that must be executed at the last stage:
\[
    L=\sum_{j=0}^{M-1}\bigl(F_{\text{last}}^j+B_{\text{last}}^j\bigr).
\]
This quantity lower bounds the makespan of any valid schedule, offline or online, because every schedule must execute all last-stage forward and backward computations. When the last stage is the dominant bottleneck, this lower bound is close to the relevant amount of unavoidable work, leading to the following approximation-style characterization.

\begin{corollary}
    \label{cor:last-stage-dominance}
    Assume that for each microbatch $j$, the last stage is the bottleneck with probability at least $1-p$, and otherwise the deviation from the last-stage time is bounded by a constant factor $\rho\ge1$. Further assume that each microbatch execution time is bounded by constants $m_l$ and $m_h$. Then
    \[
        \mathbb{E}\!\left[\frac{\mathcal{C}}{\mathrm{OPT}}\right]
        \le 1+2p(\rho-1)+O\!\left(\frac{N}{M}\right).
    \]
\end{corollary}

The corollary suggests that the arbitration rule remains close to optimal when the last stage is often the bottleneck. This matches our measured workloads: across 100 iterations, the last stage is the bottleneck for 96.8\% of forward microbatches in the LLM setting and 85.9\% in the multimodal setting. These measurements explain why the BF hint is a reasonable default for our workloads. The full system, including communication, tensor parallelism, and runtime overhead, is evaluated in the end-to-end experiments.

\section{Experiments}

\subsection{Experimental Setup}
\label{sec:7.1}
\paragraph{System setup.}
RRFP extends a Megatron-based training framework. Our experiments run on up to 128 NVIDIA RTX 4090 GPUs across 16 nodes, with 8 GPUs per node. GPUs within each node communicate over PCIe, while nodes are connected through InfiniBand. RRFP uses shared memory for intra-node transfers and RDMA for inter-node transfers. All experiment runs use FP16 precision.

\paragraph{Workloads.} We evaluate RRFP on both language-only and multimodal training workloads. The representative suite includes \textbf{LLM} (GPT3-Large~\cite{brown2020languagemodelsfewshotlearners,kalyan2023surveygpt3familylarge}), \textbf{Mid-scale LMM} (Qwen3-1.7B + ViT-H~\cite{yang2025qwen3technicalreport,DBLP:journals/corr/abs-2010-11929}), and \textbf{Heavy LMM} (Qwen3-4B + ViT-Big~\cite{yang2025qwen3technicalreport,DBLP:journals/corr/abs-2010-11929}). We additionally evaluate larger multimodal workloads in the large-scale experiments, including LLaMA3-8B + ViT-5B~\cite{dehghani2023scalingvisiontransformers22}, Qwen3-32B + InternViT~\cite{chen2024far,chen2024internvl,gao2024mini,chen2024expanding}, and LLaMA3-70B + ViT-22B~\cite{dehghani2023scalingvisiontransformers22}. We denote parallel configurations as TP/PP/DP, with total GPU count $\mathrm{TP}\times\mathrm{PP}\times\mathrm{DP}$. The exact TP/PP/DP configurations and global batch sizes used in each experiment are specified in Sections~\ref{sec:7.3}--\ref{sec:7.8}.

\paragraph{Baselines.}
For same-codebase comparison, we evaluate \textbf{1F1B}~\cite{narayanan2021memoryefficientpipelineparalleldnntraining}, \textbf{ZeroBubble (ZB)}~\cite{qi2023zerobubblepipelineparallelism}, \textbf{RRFP with the default BF hint (RRFP)}, and \textbf{RRFP with the BFW hint variant (RRFP+BFW)}. \textbf{RRFP} uses the backward-forward (BF) hint from Sections~\ref{sec:5} and~\ref{sec:6}: after readiness filtering, each round considers backward-ready work first. \textbf{RRFP+BFW} keeps the same runtime, but augments the task set with backward/weight-update decomposition and uses the resulting B/F/W hint order. This variant is inspired by ZeroBubble's backward/weight-update separation, but still treats the hint order as a non-binding preference over currently ready tasks. When neither backward nor forward computation is ready, RRFP+BFW dispatches an available weight-update task and then returns to the next arbitration round. In both RRFP and RRFP+BFW, unavailable tasks are skipped rather than waited on. For cross-framework comparison, we additionally compare with \textbf{DeepSpeed}~\cite{rajbhandari2020zeromemoryoptimizationstraining,aminabadi2022deepspeedinferenceenablingefficient} and \textbf{Cornstarch}~\cite{jang2025efficientdistributedmllmtraining}. DeepSpeed is a widely used distributed training framework with mature pipeline-parallel support, while Cornstarch is a recent multimodal training system with modality-aware pipeline scheduling.

\paragraph{Metrics.}
We use iteration time as the primary metric across experiments. All measurements discard the first iteration and start timing from the second iteration. In RQ1, we also report throughput for completeness. For RQ2, we collect compute time, blocking time, and tensor-parallel coordination overhead for runtime breakdown. Experiment-specific details are reported in the corresponding subsections.

\paragraph{Correctness validation.}
We validate training correctness by comparing RRFP with the corresponding 1F1B baseline under matched seeds and data order. RRFP follows the same loss trend without divergence, confirming that readiness-driven execution changes only runtime ordering while preserving training semantics. Detailed loss curves are provided in Appendix~\ref{app:loss-curves}.

\subsection{Evaluation Questions}

We evaluate RRFP using six research questions: \textbf{RQ1} end-to-end performance at representative and large scale, \textbf{RQ2} runtime breakdown, \textbf{RQ3} comparison with external systems, \textbf{RQ4} robustness to runtime variability, \textbf{RQ5} sensitivity to hint orders, and \textbf{RQ6} scaling across conditions. We first establish end-to-end gains in both representative and large-scale distributed settings (RQ1), then explain where the gains come from (RQ2), compare against strong external systems (RQ3), and finally validate robustness, hint-order sensitivity, and scaling behavior (RQ4--RQ6).

\subsection{RQ1: End-to-End Performance}
\label{sec:7.3}

\paragraph{Protocol.} RQ1 evaluates end-to-end training performance under two complementary settings. First, we evaluate all three representative workloads under four parallel configurations: TP1/PP8/DP1, TP1/PP16/DP1, TP2/PP8/DP1, and TP2/PP16/DP1. This yields 12 configurations. For each row, we compare 1F1B, ZeroBubble (ZB), RRFP, and RRFP+BFW. Global batch sizes are matched across methods within each row: GPT3-Large uses 64; Qwen3-1.7B + ViT-H uses 192; and Qwen3-4B + ViT-Big uses 192 at TP=1 and 96 at TP=2. Second, we evaluate 8 selected large-scale multimodal configurations on 32, 64, and 128 GPUs to validate that the gains persist under larger distributed deployments. These large-scale settings use batch size 64 and include both single-pipeline and DP-replicated hybrid-parallel configurations.

For both representative and large-scale evaluations, each data point is averaged over 10 independent runs, with 50 measured iterations per run. Each table cell reports $a/b$ ($x\times$), where $a$ is iteration time in seconds, $b$ is training throughput in s, and $x\times$ is speedup relative to 1F1B. Throughput is computed as batch size divided by iteration time. Batch size is matched across methods within each configuration.

\begin{table*}[t]
    \centering
    \caption{\textbf{RQ1: end-to-end performance on representative settings.} All settings use $\mathrm{DP}=1$. Cells report iteration time / throughput, with speedup over 1F1B shown in parentheses.}
    \label{tab:main}
    \footnotesize
    \begin{tabular}{llcccc}
        \toprule
        Workload & Parallelism  & 1F1B                       & ZB                         & RRFP                        & RRFP+BFW                    \\
        \midrule
        \multirow{4}{*}{GPT3-Large}
                 & TP1/PP8/DP1  & 3.28/19.51 (1.00$\times$)  & 2.97/21.55 (1.10$\times$)  & 2.29/27.95 (1.43$\times$) & 2.10/30.48 (1.56$\times$) \\
                 & TP1/PP16/DP1 & 2.09/30.62 (1.00$\times$)  & 1.74/36.78 (1.20$\times$)  & 1.26/50.79 (1.66$\times$) & 1.18/54.24 (1.77$\times$) \\
                 & TP2/PP8/DP1  & 3.66/17.49 (1.00$\times$)  & 3.27/19.57 (1.12$\times$)  & 3.03/21.12 (1.21$\times$) & 2.81/22.78 (1.30$\times$) \\
                 & TP2/PP16/DP1 & 2.34/27.35 (1.00$\times$)  & 2.03/31.53 (1.15$\times$)  & 1.63/39.26 (1.44$\times$) & 1.57/40.76 (1.49$\times$) \\
        \midrule
        \multirow{4}{*}{Qwen3-1.7B + ViT-H}
                 & TP1/PP8/DP1  & 11.49/16.71 (1.00$\times$) & 9.91/19.37 (1.16$\times$)  & 5.98/32.11 (1.92$\times$) & 5.68/33.80 (2.02$\times$) \\
                 & TP1/PP16/DP1 & 9.04/21.24 (1.00$\times$)  & 7.40/25.95 (1.22$\times$)  & 4.50/42.67 (2.01$\times$) & 4.33/44.34 (2.09$\times$) \\
                 & TP2/PP8/DP1  & 12.94/14.84 (1.00$\times$) & 11.86/16.19 (1.09$\times$) & 7.66/25.07 (1.69$\times$) & 7.41/25.91 (1.75$\times$) \\
                 & TP2/PP16/DP1 & 12.54/15.31 (1.00$\times$) & 11.95/16.07 (1.05$\times$) & 7.07/27.16 (1.77$\times$) & 6.74/28.49 (1.86$\times$) \\
        \midrule
        \multirow{4}{*}{Qwen3-4B + ViT-Big}
                 & TP1/PP8/DP1  & 14.05/13.67 (1.00$\times$) & 12.61/15.23 (1.11$\times$) & 7.84/24.49 (1.79$\times$) & 7.52/25.53 (1.87$\times$) \\
                 & TP1/PP16/DP1 & 11.01/17.44 (1.00$\times$) & 9.98/19.24 (1.10$\times$)  & 5.02/38.25 (2.19$\times$) & 4.81/39.92 (2.29$\times$) \\
                 & TP2/PP8/DP1  & 7.19/13.35 (1.00$\times$)  & 6.60/14.55 (1.09$\times$)  & 4.45/21.57 (1.62$\times$) & 4.19/22.91 (1.72$\times$) \\
                 & TP2/PP16/DP1 & 6.29/15.26 (1.00$\times$)  & 5.99/16.03 (1.05$\times$)  & 2.94/32.65 (2.14$\times$) & 2.81/34.16 (2.24$\times$) \\
        \bottomrule
    \end{tabular}
\end{table*}

\paragraph{Representative results.}
Across all 12 representative configurations, RRFP consistently outperforms 1F1B and ZB, and RRFP+BFW further improves performance in most settings. The results show that making execution depend on realized readiness can recover utilization that fixed-order execution cannot always capture under runtime variability.

Comparing RRFP and RRFP+BFW highlights the role of backward decomposition. RRFP already surpasses both 1F1B and ZB, showing that even the simple BF hint can be effective when applied after readiness filtering rather than enforced as a fixed execution sequence. At the same time, RRFP+BFW provides additional gains across almost all settings, indicating that RRFP can leverage the backward/weight-update separation idea when available. This suggests that RRFP is not tied to a specific schedule, but instead provides a flexible runtime layer for incorporating finer-grained execution hints as non-binding preferences over ready work.

The benefits are greater in multimodal workloads. Compared to GPT3-Large, both Qwen3-1.7B + ViT-H and Qwen3-4B + ViT-Big exhibit larger improvements, which is consistent with increased runtime variability due to heterogeneous computation across stages. Under such conditions, fixed execution orders become less reliable, while RRFP dynamically adapts to changing readiness and better utilizes idle stages.

We also observe that gains remain strong across both pipeline-depth and tensor-parallel settings. Under $\mathrm{TP}=1$, increasing pipeline depth from $\mathrm{PP}=8$ to $\mathrm{PP}=16$ generally provides more opportunities for runtime adaptation. The $\mathrm{TP}=2$ results show that this benefit persists when tensor-parallel coordination is enabled: RRFP still outperforms 1F1B and ZB under both $\mathrm{PP}=8$ and $\mathrm{PP}=16$, although the gain can be smaller than in TP-free settings due to increased intra-stage synchronization and reduced scheduling flexibility.

\begin{table*}[t]
    \centering
    \caption{\textbf{RQ1: end-to-end performance under large-scale distributed settings.} All settings use global batch size 64. Cells report iteration time / throughput, with speedup over 1F1B shown in parentheses.}
    \label{tab:large-scale-main}
    \footnotesize
    \begin{tabular}{llccccc}
        \toprule
        GPUs & Workload              & Parallelism  & 1F1B                      & ZB                        & RRFP                        & RRFP+BFW                    \\
        \midrule
        32   & LLaMA3-8B + ViT-5B    & TP1/PP32/DP1 & 7.62/8.40 (1.00$\times$)  & 6.58/9.73 (1.16$\times$)  & 3.38/18.93 (2.25$\times$) & 2.75/23.27 (2.77$\times$) \\
        32   & LLaMA3-8B + ViT-5B    & TP2/PP16/DP1 & 8.30/7.71 (1.00$\times$)  & 7.84/8.16 (1.06$\times$)  & 3.52/18.18 (2.36$\times$) & 3.04/21.05 (2.73$\times$) \\
        32   & LLaMA3-8B + ViT-5B    & TP2/PP8/DP2  & 4.90/13.06 (1.00$\times$) & 4.44/14.41 (1.10$\times$) & 3.02/21.19 (1.62$\times$) & 2.77/23.10 (1.77$\times$) \\
        64   & Qwen3-32B + InternViT & TP1/PP64/DP1 & 8.03/7.97 (1.00$\times$)  & 7.35/8.71 (1.09$\times$)  & 4.46/14.35 (1.80$\times$) & 4.37/14.65 (1.84$\times$) \\
        64   & Qwen3-32B + InternViT & TP2/PP32/DP1 & 8.39/7.63 (1.00$\times$)  & 7.66/8.36 (1.10$\times$)  & 5.19/12.33 (1.62$\times$) & 4.66/13.73 (1.80$\times$) \\
        64   & Qwen3-32B + InternViT & TP2/PP16/DP2 & 5.14/12.45 (1.00$\times$) & 4.90/13.06 (1.05$\times$) & 4.06/15.76 (1.27$\times$) & 3.74/17.11 (1.37$\times$) \\
        128  & LLaMA3-70B + ViT-22B  & TP2/PP64/DP1 & 8.72/7.34 (1.00$\times$)  & 7.94/8.06 (1.10$\times$)  & 5.16/12.40 (1.69$\times$) & 4.76/13.45 (1.83$\times$) \\
        128  & LLaMA3-70B + ViT-22B  & TP2/PP32/DP2 & 7.23/8.85 (1.00$\times$)  & 6.67/9.60 (1.08$\times$)  & 4.71/13.59 (1.54$\times$) & 4.63/13.82 (1.56$\times$) \\
        \bottomrule
    \end{tabular}
\end{table*}

\paragraph{Large-scale results.}
To test whether the gains persist beyond representative configurations, we evaluate RRFP on larger multimodal workloads using 32, 64, and 128 GPUs. Table~\ref{tab:large-scale-main} covers both deep single-pipeline settings and hybrid PP/TP/DP settings with data-parallel replication. These results validate RRFP under larger distributed deployments and different parallelization choices.

Across these large-scale settings from 32 to 128 GPUs, RRFP preserves the same qualitative gains observed in representative configurations. As workloads become larger and more heterogeneous, deeper pipelines and cross-stage dependencies amplify the cost of rigid execution order. RRFP reduces this cost by dispatching currently ready work instead of waiting for unavailable tasks, while still using scheduling hints to rank ready candidates. The DP-replicated settings further show that RRFP is not limited to a single pipeline instance: even when the GPU allocation is split across data-parallel replicas, RRFP and RRFP+BFW continue to improve iteration time over fixed-order baselines.

Overall, RQ1 demonstrates that RRFP is robust to runtime variability and compatible with finer-grained execution hints such as backward/weight-update separation. The representative results show consistent gains across workload types and TP/PP trade-offs. RRFP+BFW further shows that the same readiness-driven runtime can incorporate ZeroBubble's backward and weight-update separation idea, often improving over the default BF hint. The large-scale validation confirms that these gains extend to practical hybrid-parallel deployments. Together, these results indicate that readiness-driven execution is effective both for the default BF policy and for finer-grained variants at larger scale.

\subsection{RQ2: Runtime Breakdown}

\paragraph{Breakdown setup.} RQ2 determines where RRFP's end-to-end gains come from by decomposing each training iteration into major runtime components. In particular, we distinguish compute-related time, blocking time, and tensor-parallel coordination overhead. This lets us separate reductions in pipeline bubbles and communication/blocking stalls from changes in compute execution, while also quantifying the extra coordination cost introduced by out-of-order execution under tensor parallelism. We use Qwen3-4B + ViT-Big with batch size 32, fixed PP=16 and DP=1, varying TP over $\{1,2,4\}$. TP=1 does not require tensor-parallel collective-order coordination, while TP$>$1 allows us to measure the coordination overhead that RRFP introduces to preserve collective consistency under out-of-order execution.

\paragraph{Breakdown metrics.}
We report iteration time and three runtime components:
\begin{itemize}
    \item \emph{Compute}: total F/B computation time across all microbatches, including GPU kernel execution and CPU-side overhead such as kernel launch.
    \item \emph{Blocking}: time without useful F/B computation. For 1F1B, this is time spent issuing or waiting for point-to-point communication, including \texttt{wait\_on\_req}. For RRFP, this is time spent polling or waiting when no executable microbatch is currently ready.
    \item \emph{TP Coord}: extra time spent in RRFP's tensor-parallel metadata all-gather for preserving collective-order consistency. It is zero for 1F1B, whose fixed order already aligns collective calls, and for TP=1.
\end{itemize}

\begin{table*}[t]
    \centering
    \caption{\textbf{RQ2: runtime breakdown.} We evaluate 1F1B and RRFP on Qwen3-4B + ViT-Big with global batch size 32 and PP=16, varying TP over $\{1,2,4\}$. Values are seconds per iteration. Percentages are normalized by iteration time. Results are averaged over 3 independent runs with 10 measured iterations per run.}
    \label{tab:breakdown}
    \small
    \setlength{\tabcolsep}{4.5pt}
    \begin{tabular}{llrrrr}
        \toprule
        Parallelism & Method & Iter  & Compute         & Blocking        & TP Coord       \\
        \midrule
        \multirow{2}{*}{TP1/PP16/DP1}
                    & 1F1B   & 2.581 & 0.859 (33.28\%) & 1.722 (66.72\%) & 0.000 (0.00\%) \\
                    & RRFP     & 1.788 & 1.020 (57.05\%) & 0.768 (42.95\%) & 0.000 (0.00\%) \\
        \midrule
        \multirow{2}{*}{TP2/PP16/DP1}
                    & 1F1B   & 2.644 & 0.946 (35.78\%) & 1.698 (64.22\%) & 0.000 (0.00\%) \\
                    & RRFP     & 1.824 & 1.026 (56.25\%) & 0.788 (43.20\%) & 0.010 (0.55\%) \\
        \midrule
        \multirow{2}{*}{TP4/PP16/DP1}
                    & 1F1B   & 2.750 & 0.938 (34.11\%) & 1.812 (65.89\%) & 0.000 (0.00\%) \\
                    & RRFP     & 1.942 & 1.054 (54.27\%) & 0.873 (44.95\%) & 0.015 (0.77\%) \\
        \bottomrule
    \end{tabular}
\end{table*}

Table~\ref{tab:breakdown} reports the resulting breakdown. Across all TP settings, the dominant cost in fixed-order 1F1B is \emph{Blocking} rather than \emph{Compute}. \emph{Blocking} accounts for 64.22--66.72\% of 1F1B iteration time, indicating that a substantial fraction of time is spent without useful compute progress. This blocking arises because 1F1B follows a pre-committed execution order: each stage issues communication and computation according to the fixed schedule, and must wait when the scheduled dependency has not yet arrived. Under runtime variability, the actual completion times of computation and communication differ across stages and microbatches, so the fixed order can become misaligned with realized readiness, creating pipeline bubbles and communication stalls.

RRFP reduces this cost through readiness-driven out-of-order execution. Instead of waiting for the next task in a fixed global order, each stage executes ready work as it becomes available. As a result, \emph{Blocking} is consistently reduced across TP settings: from 1.722s to 0.768s under TP1/PP16/DP1, from 1.698s to 0.788s under TP2/PP16/DP1, and from 1.812s to 0.873s under TP4/PP16/DP1. This reduction is the primary source of the end-to-end speedup.

This breakdown also shows that the speedup is not due to faster computation or kernel-level acceleration. \emph{Compute} remains comparable between 1F1B and RRFP, and is slightly higher for RRFP in these runs due to additional CPU-side runtime overhead around computation.

As TP increases from 1 to 4, RRFP introduces \emph{TP Coord} overhead to preserve tensor-parallel collective-order consistency under out-of-order execution. This overhead remains small: 0.010s under TP=2 and 0.015s under TP=4, accounting for less than 1\% of iteration time. Thus, the coordination cost is much smaller than the reduction in \emph{Blocking}.

Overall, RQ2 shows that RRFP shifts the cost structure of pipeline execution: it adds lightweight runtime polling and TP coordination, but substantially reduces blocking caused by pre-committed execution order. This supports the main runtime-system claim that readiness-driven progress reduces idle time under runtime variability.

\subsection{RQ3: Cross-Framework Performance}

\paragraph{Cross-framework setup.}
RQ3 compares the default RRFP configuration (RRFP) with DeepSpeed and Cornstarch on the same matched RQ1 settings. We use RRFP to isolate the readiness-driven runtime from the optional backward/weight-update decomposition in RRFP+BFW. We evaluate all 12 representative configurations in Table~\ref{tab:main} and all 8 large-scale configurations in Table~\ref{tab:large-scale-main}. For $\mathrm{TP}>1$, we compare only Cornstarch and RRFP because DeepSpeed's automatic tensor-parallel partitioning is not applicable to our PipelineModule-wrapped model in this training setup~\cite{aminabadi2022deepspeedinferenceenablingefficient,rajbhandari2020zeromemoryoptimizationstraining}.

\begin{table*}[t]
    \centering
    \caption{\textbf{Cross-framework comparison on representative settings.}
        Settings match the corresponding RQ1 configurations. Entries report average iteration time in seconds, and speedup is computed over the \textbf{faster} available external baseline.}
    \label{tab:external}
    \small
    \setlength{\tabcolsep}{4.2pt}
    \begin{tabular}{llcccc}
        \toprule
        Workload & Parallelism  & DeepSpeed & Cornstarch & RRFP   & Speedup      \\
        \midrule
        \multirow{4}{*}{GPT3-Large}
                 & TP1/PP8/DP1  & 2.60      & 2.63       & 2.29 & 1.14$\times$ \\
                 & TP1/PP16/DP1 & 1.52      & 2.00       & 1.26 & 1.21$\times$ \\
                 & TP2/PP8/DP1  & N/A       & 3.25       & 3.03 & 1.07$\times$ \\
                 & TP2/PP16/DP1 & N/A       & 2.88       & 1.63 & 1.77$\times$ \\
        \midrule
        \multirow{4}{*}{Qwen3-1.7B + ViT-H}
                 & TP1/PP8/DP1  & 9.73      & 8.65       & 5.98 & 1.45$\times$ \\
                 & TP1/PP16/DP1 & 7.80      & 6.71       & 4.50 & 1.49$\times$ \\
                 & TP2/PP8/DP1  & N/A       & 9.88       & 7.66 & 1.29$\times$ \\
                 & TP2/PP16/DP1 & N/A       & 8.16       & 7.07 & 1.15$\times$ \\
        \midrule
        \multirow{4}{*}{Qwen3-4B + ViT-Big}
                 & TP1/PP8/DP1  & 11.38     & 10.12      & 7.84 & 1.29$\times$ \\
                 & TP1/PP16/DP1 & 10.54     & 9.24       & 5.02 & 1.84$\times$ \\
                 & TP2/PP8/DP1  & N/A       & 6.80       & 4.45 & 1.53$\times$ \\
                 & TP2/PP16/DP1 & N/A       & 4.76       & 2.94 & 1.62$\times$ \\
        \bottomrule
    \end{tabular}
\end{table*}

\paragraph{Representative results.}
Table~\ref{tab:external} shows that RRFP achieves the lowest iteration time across all 12 representative settings in cross-framework comparison, with $1.07\times$--$1.84\times$ speedup over the faster available external baseline. As these settings match RQ1, the results show that RRFP's gains persist across workload type, pipeline depth, and tensor parallelism. The gains are generally larger on multimodal workloads than on GPT3-Large, consistent with greater stage heterogeneity and runtime variability in multimodal training.

\begin{table*}[t]
    \centering
    \caption{\textbf{Large-scale cross-framework comparison.}
        Settings match the corresponding RQ1 configurations. Entries report average iteration time in seconds, and speedup is computed over the \textbf{faster} available external baseline.}
    \label{tab:external-large}
    \small
    \setlength{\tabcolsep}{4.0pt}
    \begin{tabular}{llccccc}
        \toprule
        GPUs & Workload              & Parallelism  & DeepSpeed & Cornstarch & RRFP   & Speedup      \\
        \midrule
        32   & LLaMA3-8B + ViT-5B    & TP1/PP32/DP1 & 4.33      & 3.42       & 3.38 & 1.01$\times$ \\
        32   & LLaMA3-8B + ViT-5B    & TP2/PP16/DP1 & N/A       & 4.40       & 3.52 & 1.25$\times$ \\
        32   & LLaMA3-8B + ViT-5B    & TP2/PP8/DP2  & N/A       & 3.51       & 3.02 & 1.16$\times$ \\
        64   & Qwen3-32B + InternViT & TP1/PP64/DP1 & 5.97      & 5.51       & 4.46 & 1.24$\times$ \\
        64   & Qwen3-32B + InternViT & TP2/PP32/DP1 & N/A       & 5.67       & 5.19 & 1.09$\times$ \\
        64   & Qwen3-32B + InternViT & TP2/PP16/DP2 & N/A       & 4.30       & 4.06 & 1.06$\times$ \\
        128  & LLaMA3-70B + ViT-22B  & TP2/PP64/DP1 & N/A       & 5.52       & 5.16 & 1.07$\times$ \\
        128  & LLaMA3-70B + ViT-22B  & TP2/PP32/DP2 & N/A       & 5.79       & 4.71 & 1.23$\times$ \\
        \bottomrule
    \end{tabular}
\end{table*}

\paragraph{Large-scale results.} Table~\ref{tab:external-large} shows that RRFP also achieves the lowest iteration time across all 8 large-scale settings, with $1.01\times$--$1.25\times$ speedup over the faster available external baseline. The smallest gap appears for LLaMA3-8B + ViT-5B under TP1/PP32/DP1, where Cornstarch is already competitive. The remaining cases show clearer gains, including settings with tensor parallelism and data-parallel replication, indicating that RRFP remains effective as training scales to deeper pipelines and hybrid TP/PP/DP configurations.

Overall, RQ3 shows that RRFP's gains extend beyond same-codebase schedule comparisons and hold against external systems. Across the representative and large-scale RQ1 settings, RRFP improves over both general-purpose and multimodal-specialized external training systems. This suggests that RRFP's advantage comes from the readiness-driven runtime model rather than artifacts of a particular framework or schedule implementation.

\subsection{RQ4: Robustness to Runtime Variability}

\paragraph{Robustness under variability.}
We test how much performance degrades when runtime jitter increases under a controlled injection setup. We use Qwen3-4B + ViT-Big with TP2/\allowbreak PP8/\allowbreak DP1 and global batch size 96, compare 1F1B and RRFP, and keep all settings fixed except the injected jitter level.

\paragraph{Jitter model and protocol.}
We inject compute-path jitter on every pipeline stage by adding random CUDA-side sleep delays to F/B compute tasks. For each task, the runtime first measures its compute time $c_t$ and updates a stage-local exponential moving average:
\[
    e_t = 0.9e_{t-1}+0.1c_t .
\]
With probability $p_j$, we inject a delay
\[
    d_t=\alpha \max(B,e_t)(0.5+r_t),\quad r_t\sim\mathrm{Uniform}(0,1),
\]
where $B$ sets the base delay and $\alpha$ controls the delay scale. We evaluate four jitter levels, $J_0$--$J_3$, with parameters shown in Table~\ref{tab:robustness}. For each level, 1F1B and RRFP use the same random seed, so they experience paired jitter patterns. Each setting is repeated 3 times with 10 measured iterations per run. We report the mean iteration time, slowdown relative to each method's no-injection baseline, and the standard deviation computed over the 30 measured iterations.

\begin{table}[t]
    \centering
    \caption{\textbf{RQ4: robustness under compute-path jitter.} We evaluate 1F1B and RRFP on Qwen3-4B + ViT-Big with TP2/PP8/DP1 and global batch size 96. Parentheses show slowdown relative to each method's no-injection baseline.}
    \label{tab:robustness}
    \scriptsize
    \setlength{\tabcolsep}{3.0pt}
    \resizebox{\columnwidth}{!}{%
        \begin{tabular}{lccc|cc|cc}
            \toprule
            \multirow{2}{*}{Level}
                  & \multirow{2}{*}{$p_j$}
                  & \multirow{2}{*}{$B$ (ms)}
                  & \multirow{2}{*}{$\alpha$}
                  & \multicolumn{2}{c|}{1F1B}
                  & \multicolumn{2}{c}{RRFP}                                                                      \\
            \cmidrule(lr){5-6}\cmidrule(lr){7-8}
                  &                           &         &
                  & $T$ (s)                   & std (s)
                  & $T$ (s)                   & std (s)                                                         \\
            \midrule
            $J_0$ & 0                         & 0       & 0.0 & 7.19 (0.0\%)   & 0.339 & 4.45 (0.0\%)   & 0.191 \\
            $J_1$ & 0.1                       & 5       & 0.5 & 7.39 (2.82\%)  & 0.325 & 4.53 (1.81\%)  & 0.191 \\
            $J_2$ & 0.2                       & 10      & 1.0 & 8.00 (11.27\%) & 0.410 & 4.75 (6.82\%)  & 0.210 \\
            $J_3$ & 0.3                       & 15      & 1.5 & 8.49 (18.06\%) & 0.390 & 4.96 (11.36\%) & 0.228 \\
            \bottomrule
        \end{tabular}
    }
\end{table}

\paragraph{Results and discussion.}
Table~\ref{tab:robustness} shows that RRFP consistently remains faster than 1F1B across all jitter levels. More importantly, RRFP degrades more slowly as jitter increases: from $J_1$ to $J_3$, RRFP's slowdown grows from 1.81\% to 11.36\%, whereas 1F1B grows from 2.82\% to 18.06\%. This indicates that readiness-driven execution is less sensitive to compute-path perturbations than a fixed execution order.

RRFP also shows lower variation in measured iteration time across all jitter levels. Its standard deviation remains between 0.191s and 0.228s, compared with 0.325s to 0.410s for 1F1B. Overall, RQ4 shows that RRFP is robust to runtime variability: it maintains lower iteration time across all tested jitter levels and degrades more slowly as jitter intensity increases.

\subsection{RQ5: Sensitivity to Hint Orders}
\label{sec:7.7}

\paragraph{Hint-order sensitivity.}
We evaluate whether RRFP depends on a particular ready-set hint order. All variants run inside the same RRFP runtime loop and differ only in how they rank currently ready forward and backward candidates, while using the same readiness filtering. We compare four hint orders. \emph{BF} is the default backward-forward hint from Section~\ref{sec:5}: each arbitration round first considers backward-ready work and then forward-ready work, subject to runtime readiness. \emph{FB} reverses this order by considering forward-ready work before backward-ready work. \emph{B-priority} selects backward work whenever any backward candidate is ready, while \emph{F-priority} analogously prioritizes forward work whenever any forward candidate is ready.

\paragraph{Configuration and metric.}
We use Qwen3-1.7B + ViT-H with TP1/PP8/DP1, global batch size 192, and 10 measured training iterations per run. Each hint order is repeated 3 times. We report average end-to-end iteration time.

\begin{table}[t]
    \centering
    \caption{\textbf{RQ5: hint-order sensitivity}. We evaluate RRFP on Qwen3-1.7B + ViT-H with TP1/PP8/DP1 and batch size 192.}
    \label{tab:heuristic}
    \setlength{\tabcolsep}{4pt}
    \begin{tabular}{lcc}
        \toprule
        Hint order   & Avg (ms) & Slowdown (\%) \\
        \midrule
        BF (default) & 6007.26  & 0.00          \\
        FB           & 6034.10  & +0.45         \\
        B-priority   & 6057.79  & +0.84         \\
        F-priority   & 6302.94  & +4.92         \\
        \bottomrule
    \end{tabular}
\end{table}

\paragraph{Results and discussion.}
Table~\ref{tab:heuristic} shows that RRFP is largely insensitive to the specific hint order across these variants. BF, FB, and B-priority remain within 1\% of the default, indicating that performance is not driven by a carefully tuned ordering rule but primarily by the underlying readiness-driven execution model. At the same time, poorly aligned hint orders can still hurt performance: F-priority incurs a noticeable slowdown of +4.92\% because it keeps prioritizing forward work whenever possible, delaying backward execution and creating additional pipeline bubbles. This shows that hint orders still affect how RRFP chooses among ready tasks, but the small gaps among BF, FB, and B-priority indicate that the main benefit comes from readiness-driven execution rather than from a specific heuristic.

Overall, RQ5 reinforces that RRFP's gains mainly come from out-of-order, readiness-driven execution rather than from a specific heuristic. This is consistent with RQ1 and RQ2, where RRFP improves utilization by reducing blocking wait, and shows that RRFP acts as a flexible runtime layer rather than a heuristic-specific scheduler.

\subsection{RQ6: Scaling Across Conditions}
\label{sec:7.8}
\paragraph{Setup.}
RQ6 analyzes when RRFP helps most along three controlled axes: pipeline depth, multimodal imbalance, and global batch size. For pipeline-depth scaling, we vary PP while keeping the workload, TP/DP configuration, and global batch size fixed. For modality-imbalance scaling, we fix the language backbone and parallel configuration while increasing the vision encoder size. In the ViT rows of Table~\ref{tab:scaling}, L/H/g/Big denote ViT-L, ViT-H, ViT-g, and ViT-Big, and 5B/Intern/22B denote ViT-5B, InternViT, and ViT-22B. For batch-size scaling, we vary global batch size under two fixed workload settings. Table~\ref{tab:scaling} lists the concrete workloads, parallel configuration, sweeps, and global batch sizes. Each configuration is run 10 times with 50 measured iterations per run, and we report speedup over 1F1B.

\begin{table*}[t]
    \centering
    \caption{\textbf{RQ6 scaling results across pipeline depth, modality imbalance, and global batch size.} We report speedup over 1F1B. Larger values indicate greater benefit from readiness-driven execution.}
    \label{tab:scaling}
    \footnotesize
    \setlength{\tabcolsep}{4pt}
    \begin{tabular}{llllcc}
        \toprule
        Axis & Workload           & Config       & Sweep             & Batch size & Speedup                   \\
        \midrule
        PP   & Qwen3-1.7B + ViT-H & TP1/DP1      & PP4 / PP8 / PP16  & 192        & 1.07 / 1.94 / 2.01        \\
        PP   & LLaMA3-8B + ViT-5B & TP1/DP1      & PP8 / PP16 / PP32 & 64         & 1.79 / 2.13 / 2.31        \\
        PP   & Qwen3-4B + ViT-Big & TP2/DP1      & PP4 / PP8 / PP16  & 64         & 1.03 / 1.64 / 1.72        \\
        \midrule
        ViT  & Qwen3-1.7B         & TP1/PP16/DP1 & L / H / g / Big   & 192        & 1.73 / 1.89 / 2.04 / 2.16 \\
        ViT  & LLaMA3-8B          & TP1/PP32/DP1 & 5B / Intern / 22B & 64         & 1.32 / 1.43 / 1.78        \\
        \midrule
        BSZ  & Qwen3-4B + ViT-Big & TP1/PP16/DP1 & 64 / 128 / 192    & --         & 1.98 / 2.14 / 2.18        \\
        BSZ  & LLaMA3-8B + ViT-5B & TP2/PP32/DP1 & 64 / 128 / 192    & --         & 2.62 / 2.72 / 2.82        \\
        \bottomrule
    \end{tabular}
\end{table*}

\paragraph{Results and discussion.}
Table~\ref{tab:scaling} shows that RRFP's benefit increases when execution provides more opportunities for readiness-driven choice. For pipeline-depth scaling, speedup grows from 1.07$\times$ to 2.01$\times$ on Qwen3-1.7B + ViT-H as PP increases from 4 to 16, and from 1.79$\times$ to 2.31$\times$ on LLaMA3-8B + ViT-5B as PP increases from 8 to 32. The TP2 Qwen3-4B + ViT-Big setting follows the same trend, with speedup increasing from 1.03$\times$ to 1.72$\times$ as pipeline depth increases, indicating that the benefit consistently persists when tensor-parallel coordination is enabled.

For modality imbalance, speedup increases as the vision encoder becomes larger and stage costs become more uneven: from 1.73$\times$ to 2.16$\times$ on Qwen3-1.7B, and from 1.32$\times$ to 1.78$\times$ on LLaMA3-8B. This suggests that RRFP is especially beneficial when multimodal workloads introduce heterogeneous stage costs, making fixed execution orders more likely to become misaligned with realized readiness.

For batch-size scaling, speedup also increases as global batch size grows across both workloads. Qwen3-4B + ViT-Big improves from 1.98$\times$ to 2.18$\times$ as global batch size increases from 64 to 192, and LLaMA3-8B + ViT-5B improves from 2.62$\times$ to 2.82$\times$. Increasing the global batch size place more microbatches in the pipeline, exposing more scheduling opportunities and runtime variation for RRFP to exploit.

Overall, RQ6 shows that RRFP helps most when fixed execution orders are more likely to become stale and more alternative executable candidates are available, such as in deeper pipelines, more imbalanced multimodal workloads, and larger batches. In these settings, selecting among ready tasks instead of waiting for the next scheduled task provides larger benefits across the tested configurations, and the gains persist when tensor-parallel coordination is enabled.
\section{Related Work}

\paragraph{Pipeline parallelism.}
Pipeline parallelism is widely used to train large models that exceed the memory capacity of a single device. By partitioning a model across devices and executing different microbatches concurrently, pipeline parallelism improves hardware utilization and training throughput. GPipe~\cite{huang2019gpipeefficienttraininggiant} introduced synchronous pipeline training with microbatching, and PipeDream~\cite{10.1145/3341301.3359646} proposed the one-forward-one-backward (1F1B) schedule to reduce bubbles by overlapping forward and backward computation. This schedule has been adopted in systems such as Megatron-LM~\cite{shoeybi2020megatronlmtrainingmultibillionparameter} and is now a standard baseline for pipeline-parallel training. Later systems further improve pipeline execution: PipeDream-2BW~\cite{narayanan2021memoryefficientpipelineparalleldnntraining} reduces memory overhead from weight versioning, while ZeroBubble~\cite{qi2023zerobubblepipelineparallelism} reduces bubbles by reordering backward-related computation. RRFP builds on these foundations but does not propose another fixed ordering of forward and backward tasks. Instead, it changes how schedules are consumed at runtime: each stage repeatedly chooses from currently ready microbatches, while existing schedules serve as non-binding hint orders.

\paragraph{Pipeline planning and schedule optimization.}
Many distributed training systems optimize pipeline execution through offline search, profiling, or cost-based planning. Representative examples include Alpa~\cite{zheng2022alpaautomatinginterintraoperator} and PipeWeaver~\cite{xue2025pipeweaveraddressingdatadynamicity}. Related strategies also appear in nnScaler~\cite{lin2024nnscaler} and Unity~\cite{UngerColin2022UADT}. These methods construct better parallelization plans or better pipeline schedules under estimated workload and system conditions, and can be effective when those estimates remain representative at runtime. However, the resulting plans still rely on specific decisions made before the actual runtime ready set is observed. RRFP is orthogonal to such planning methods. It is not a schedule optimizer and does not require an optimized decision sequence for correctness or progress. Instead, it provides a readiness-driven runtime layer that can execute independently of a planned schedule, or use an existing one as a soft ranking hint when available.

\paragraph{Runtime adaptation.}
Several systems use runtime feedback to improve distributed training efficiency. ByteScheduler~\cite{10.1145/3341301.3359642} dynamically schedules communication operations to reduce synchronization overhead, HetPipe~\cite{park2020hetpipeenablinglargednn} addresses heterogeneous GPU clusters through load-balancing strategies, and Pollux~\cite{qiao2021polluxcoadaptiveclusterscheduling} adapts cluster resource allocation based on runtime feedback. These systems show the value of runtime information, but they target different layers of the training stack, such as communication prioritization, heterogeneous placement, or cluster-level scheduling. RRFP focuses on the pipeline execution layer. It makes decisions at the granularity of forward and backward microbatch tasks within each pipeline stage, allowing execution to react directly to realized readiness, runtime slowdowns, and communication delays. Thus, RRFP differs from prior runtime-adaptive systems not by periodically adapting a fixed schedule, but by making readiness the primary execution interface.

\paragraph{Communication assumptions in training runtimes.}
Modern distributed training systems rely on high-throughput GPU communication libraries whose operations are typically issued by the application in a consistent order across participating ranks~\cite{nccl,nccl-doc}. This assumption fits static or synchronized pipeline schedules, where computation and communication follow an agreed sequence. Readiness-driven out-of-order execution weakens this assumption: different stages may finish microbatches in different orders, causing sends and receives to be triggered asynchronously and out of order. If communication remains tied to a fixed operation sequence, this divergence can introduce blocking or deadlock. RRFP addresses this with message-driven asynchronous communication: completed microbatches are packaged as messages, transferred independently of the compute thread, and inserted into per-microbatch ready buffers upon arrival. It also uses lightweight tensor-parallel coordination to preserve collective-order consistency when ranks in the same tensor-parallel group might otherwise select different ready microbatches. Together, these mechanisms allow RRFP to execute microbatches safely in an order that differs from any pre-planned schedule.

\section{Conclusion}

Pipeline parallelism is essential for scaling large-model training, but fixed execution orders can become stale under runtime variability, as computation and communication readiness change during execution. This paper presented RRFP, a readiness-driven runtime layer for out-of-order pipeline-parallel training. RRFP combines message-driven asynchronous communication, ready-set arbitration, and lightweight tensor-parallel coordination to dispatch ready work while preserving communication progress and collective-order consistency. Across language-only and multimodal workloads, RRFP reduces blocking wait, improves end-to-end performance, outperforms external systems in matched settings, remains robust to injected jitter and alternative hint orders, and is especially beneficial in deeper pipelines, more imbalanced workloads, and larger-batch settings. These results suggest that pipeline schedules should guide ready work rather than serve as fixed orders that runtime must follow.
\newpage
\bibliographystyle{ACM-Reference-Format}
\bibliography{RRFP}
\newpage
\appendix

\section{Backward-Forward Hint Algorithm}
\label{app:algo}
This section gives the detailed instantiation of the backward-forward (BF) hint used by RRFP. The BF hint is a deterministic priority rule over currently ready work. It does not prescribe a fixed global execution order. Instead, it ranks only the microbatches that are already present in the local ready buffers.

The BF hint defines the priority order used by the arbitration layer as follows.

\begin{itemize}
    \item \textbf{Direction-level priority.} Arbitration proceeds in repeated rounds. In each round, the runtime first checks whether the backward-ready buffer $\mathcal{L}^{b}_{r}$ contains executable work. If so, it selects and executes one backward microbatch. The runtime then checks whether the forward-ready buffer $\mathcal{L}^{f}_{r}$ contains executable work. If so, it selects and executes one forward microbatch. Thus, when both backward and forward work are ready, the BF hint follows a 1F1B-like order by giving backward work the first opportunity in each round and then considering forward work. When only one direction has ready work, the runtime simply executes the ready direction rather than waiting for work from the other direction.

    \item \textbf{Within-direction priority.} When multiple ready microbatches exist in the same direction, the BF hint prioritizes the microbatch that corresponds to earlier progress within the stage. For forward-ready work, this means choosing the smaller model-chunk index, and for backward-ready work, this means choosing the larger model-chunk index. This rule is useful for interleaved pipeline schedules, where each physical stage may own multiple model chunks. Prioritizing earlier-progress model chunks helps them finish earlier and enter subsequent collective operations, such as gradient reduction, earlier.

    \item \textbf{Tie breaking.} If multiple ready microbatches have the same direction and the same model-chunk index, the BF hint breaks ties by choosing the smaller microbatch index. This makes arbitration deterministic and preserves a simple in-order preference within each model chunk.
\end{itemize}

Importantly, the BF hint never blocks on unavailable work: it does not force the runtime to wait for a higher-priority but unready microbatch. If the backward-ready buffer is empty in a round, the runtime immediately considers forward-ready work. Similarly, within a direction, the hint only ranks microbatches that are already present in the corresponding ready buffer.

Algorithm~\ref{alg:1} summarizes the event-driven execution loop at stage $s$ using this BF hint. Ready microbatches are inserted into the ready buffers by the communication runtime, and the arbitration layer repeatedly selects executable microbatches from these buffers. Completed microbatches are inserted into the corresponding finished buffers for downstream communication. In the algorithm, \textsc{NextByPriority}$(\cdot,\Pi)$ implements the within-direction priority rule described above: for forward-ready work, it selects the ready microbatch with the smallest model-chunk index, and for backward-ready work, it selects the ready microbatch with the largest model-chunk index. Remaining ties are broken by the smallest microbatch index.

\begin{algorithm}[t]
	\caption{Backward-forward event-driven execution loop at stage $s$}
	\label{alg:1}
	\begin{algorithmic}[1]
		\State \textbf{State:}
		forward-ready buffer $\mathcal{L}^{f}_{r}$,
		backward-ready buffer $\mathcal{L}^{b}_{r}$,
		forward-finished buffer $\mathcal{L}^{f}_{fin}$,
		backward-finished buffer $\mathcal{L}^{b}_{fin}$
		\State \textbf{Hint order:} $\Pi$: forward uses smaller model-chunk index first; backward uses larger model-chunk index first; ties use smaller microbatch index
		\State \textbf{Target:} $N_{\text{target}}$ completed executions
		\State $n_{\text{done}} \gets 0$

		\While{$n_{\text{done}} < N_{\text{target}}$}

		\Statex \Comment{Each arbitration round first checks backward-ready work.}
		\If{$\mathcal{L}^{b}_{r} \neq \emptyset$}
		\State $m_b^\star \gets \textsc{NextByPriority}(\mathcal{L}^{b}_{r},\Pi)$
		\State \textsc{Execute}$(m_b^\star)$
		\State $n_{\text{done}} \gets n_{\text{done}} + 1$
		\State $\mathcal{L}^{b}_{fin} \gets \mathcal{L}^{b}_{fin} \cup \{m_b^\star\}$
		\EndIf

		\Statex \Comment{The same round then checks forward-ready work if more work is needed.}
		\If{$\mathcal{L}^{f}_{r} \neq \emptyset$ \textbf{and} $n_{\text{done}} < N_{\text{target}}$}
		\State $m_f^\star \gets \textsc{NextByPriority}(\mathcal{L}^{f}_{r},\Pi)$
		\State \textsc{Execute}$(m_f^\star)$
		\State $n_{\text{done}} \gets n_{\text{done}} + 1$
		\State $\mathcal{L}^{f}_{fin} \gets \mathcal{L}^{f}_{fin} \cup \{m_f^\star\}$
		\EndIf

		\EndWhile
	\end{algorithmic}
\end{algorithm}

\section{Proofs and Bottleneck Statistics}
\label{app:analysis}

This section provides the detailed proof of Theorem~6.1 and Corollary~6.2 in the main paper, together with bottleneck statistics showing how often the last-stage-dominance condition holds in our workloads. We analyze RRFP under the same simplified setting as in the main paper: a non-interleaved pipeline, computation-only latency, and no tensor-parallel coordination or implementation overhead. Communication time is ignored, while inter-stage pipeline dependencies are still respected.

\subsection{Preliminaries}
\label{app:preliminaries}

Let $F_i^j$ denote the forward computation time of microbatch $j$ at pipeline stage $i$, where $i\in\{0,1,\ldots,N-1\}$ and $j\in\{0,1,\ldots,M-1\}$. We define $B_i^j$ analogously for backward computation time. For each microbatch $j$, define
\[
    F_{\max}^j=\max_{0\le i\le N-1}F_i^j,
    \qquad
    B_{\max}^j=\max_{0\le i\le N-1}B_i^j .
\]
We also define the last-stage computation times
\[
    F_{\mathrm{last}}^j=F_{N-1}^j,
    \qquad
    B_{\mathrm{last}}^j=B_{N-1}^j .
\]

Let $\mathcal{C}$ denote the iteration makespan. Let $\mathcal{F}$ be the completion time of the forward-only pipeline, where all forward work is initially ready at stage $0$ and inter-stage forward dependencies are respected. Let $\mathcal{B}$ be the completion time of the backward-only pipeline, where all backward work is initially ready at stage $N-1$ and inter-stage backward dependencies are respected.

\subsection{Proof of Theorem~6.1}
\label{app:proof-thm-6-1}

We first bound the delay that backward work can introduce to forward execution.

\begin{lemma}
\label{lemma:forward-delay}
Let $s_f^{i,j}$ be the start time of the forward task of microbatch $j$ at stage $i$ under RRFP, and let $\tilde{s}_f^{i,j}$ be the corresponding start time in the forward-only schedule. Let $K_{i,j}$ denote the set of backward microbatches executed on stage $i$ before forward task $(i,j)$ starts under RRFP. Then
\[
    s_f^{i,j}
    \le
    \tilde{s}_f^{i,j}
    +
    \sum_{k\in K_{i,j}}B_{\max}^k .
\]
\end{lemma}

\begin{proof}
Define
\[
    D_f^{i,j}=s_f^{i,j}-\tilde{s}_f^{i,j}.
\]
We prove
\[
    D_f^{i,j}\le \sum_{k\in K_{i,j}}B_{\max}^k
\]
by induction on $i+j$.

In the forward-only schedule,
\[
\tilde{s}_f^{i,j}=
\begin{cases}
0, & i=0,\;j=0,\\
\tilde{e}_f^{i-1,0}, & i>0,\;j=0,\\
\tilde{e}_f^{0,j-1}, & i=0,\;j>0,\\
\max\{\tilde{e}_f^{i-1,j},\tilde{e}_f^{i,j-1}\}, & i>0,\;j>0,
\end{cases}
\qquad
\tilde{e}_f^{i,j}=\tilde{s}_f^{i,j}+F_i^j .
\]

If $j=0$, no backward task can execute before the first forward microbatch reaches any stage. Hence $K_{i,0}=\emptyset$ and $D_f^{i,0}=0$.

If $i=0$, stage $0$ has no upstream forward dependency. Therefore the delay of forward task $(0,j)$ under RRFP is exactly the local backward work executed before it on stage $0$, which is at most
\[
    \sum_{k\in K_{0,j}}B_{\max}^k .
\]

Now consider $i>0$ and $j>0$, and assume the claim holds for all pairs $(i',j')$ with $i'+j'<i+j$. Let
\[
    R_{i,j}=\max\{e_f^{i-1,j},e_f^{i,j-1}\}
\]
be the time at which the inter-stage and local forward predecessors of task $(i,j)$ have both completed under RRFP. If stage $i$ is idle at time $R_{i,j}$, then forward task $(i,j)$ can be dispatched immediately. Otherwise, the only reason it does not start immediately is that stage $i$ is executing a backward task that was selected before $(i,j)$ became ready. Under the BF hint, after this backward task completes, the runtime probes forward-ready work again. Since $(i,j)$ is then ready, no additional backward task can be inserted before $(i,j)$. Therefore, after both predecessors are complete, forward task $(i,j)$ can be delayed by at most one additional backward task.

Let $\Delta_{i,j}\ge0$ denote the execution time of this possible additional backward task, and set $\Delta_{i,j}=0$ if no such task exists. Then according to the above deduction
\[
    s_f^{i,j}
    \le
    \max\{e_f^{i-1,j},\,e_f^{i,j-1}\}+\Delta_{i,j},
\]
and combining with the fact that
\[
    \tilde{s}_f^{i,j}
    =
    \max\{\tilde{e}_f^{i-1,j},\,\tilde{e}_f^{i,j-1}\},
\]
we obtain
\[
\begin{aligned}
    s_f^{i,j}-\tilde{s}_f^{i,j}
    &\le
    \max\{e_f^{i-1,j},\,e_f^{i,j-1}\}-\max\{\tilde{e}_f^{i-1,j},\,\tilde{e}_f^{i,j-1}\}+\Delta_{i,j}\\
    &\le
    \max\{e_f^{i-1,j}-\tilde{e}_f^{i-1,j},\,e_f^{i,j-1}-\tilde{e}_f^{i,j-1}\}+\Delta_{i,j}\\
    &=\max\{s_f^{i-1,j}-\tilde{s}_f^{i-1,j},\,s_f^{i,j-1}-\tilde{s}_f^{i,j-1}\}+\Delta_{i,j},
\end{aligned}
\]
using the fact that
\[
e_f^{i,j}-s_f^{i,j}=\tilde{e}_f^{i,j}-\tilde{s}_f^{i,j}=F_i^j,\forall i,j
\]
and that
\[
\begin{aligned}
\max\{a,b\}-\max\{c,d\}&=\frac{a+b-c-d}{2}+\frac{|a-b|-|c-d|}{2}\\
&\le\frac{a-c+b-d}{2}+\frac{|a-b-(c-d)|}{2}\\
&=\max\{a-c,b-d\}
\end{aligned}
\]
given that $a\ge c,b\ge d$.

Therefore we arrive at
\[
    D_f^{i,j}
    \le
    \max\{D_f^{i-1,j},\,D_f^{i,j-1}\}+\Delta_{i,j}.
\]
Equivalently,
\[
    D_f^{i,j}
    \le
    \max\{D_f^{i-1,j}+\Delta_{i,j},\,D_f^{i,j-1}+\Delta_{i,j}\}.
\]
By the induction hypothesis,
\[
    D_f^{i-1,j}\le\sum_{k\in K_{i-1,j}}B_{\max}^k,
    \qquad
    D_f^{i,j-1}\le\sum_{k\in K_{i,j-1}}B_{\max}^k .
\]

It remains to account for $\Delta_{i,j}$. If $\Delta_{i,j}>0$, let $q$ be the backward microbatch executed on stage $i$ after both forward predecessors of $(i,j)$ have completed but before $(i,j)$ starts. This microbatch is not contained in $K_{i,j-1}$, because it is executed after forward task $(i,j-1)$ has completed, whereas $K_{i,j-1}$ only contains backward microbatches executed before $(i,j-1)$ starts. It is also not contained in $K_{i-1,j}$: if $q\in K_{i-1,j}$, then the backward task of $q$ on stage $i-1$ would have executed before forward task $(i-1,j)$ starts. By the backward dependency, the corresponding backward task on stage $i$ must have completed even earlier, contradicting the fact that this stage-$i$ backward task is still the extra work executed after both predecessors of $(i,j)$ are ready. Thus $q$ is not counted in either predecessor-delay term.

We next relate the predecessor-delay sets to $K_{i,j}$. Clearly, $K_{i,j-1}\subseteq K_{i,j}$, because any backward microbatch executed on stage $i$ before forward task $(i,j-1)$ starts is also executed before the later forward task $(i,j)$ starts. For $K_{i-1,j}$, the inclusion also holds by the direction of backward propagation. If $k\in K_{i-1,j}$, then the backward task of $k$ on stage $i-1$ executes before forward task $(i-1,j)$ starts. Since backward execution at stage $i-1$ depends on the corresponding backward execution at stage $i$, the backward task of $k$ on stage $i$ must have completed even earlier. Because forward task $(i,j)$ cannot start before forward task $(i-1,j)$ completes, this stage-$i$ backward task is executed before $(i,j)$ starts. Hence $k\in K_{i,j}$.

Finally, if such an extra backward task $q$ exists, then by definition it is executed on stage $i$ before forward task $(i,j)$ starts, so $q\in K_{i,j}$ and
\[
    \Delta_{i,j}\le B_{\max}^q .
\]
Therefore, the predecessor-delay terms are accounted for by elements of $K_{i,j}$, and the only additional delay $\Delta_{i,j}$ is accounted for by the new element $q\in K_{i,j}$. Since $q\notin K_{i-1,j}\cup K_{i,j-1}$, there is no double counting. Hence,
\[
    D_f^{i,j}
    \le
    \max\left\{
        \sum_{k\in K_{i-1,j}}B_{\max}^k,\,
        \sum_{k\in K_{i,j-1}}B_{\max}^k
    \right\}
    + B_{\max}^q
    \le
    \sum_{k\in K_{i,j}}B_{\max}^k .
\]
This proves the lemma.
\end{proof}

Applying Lemma~\ref{lemma:forward-delay} to stage $N-1$ and microbatch $M-1$ gives
\[
    s_f^{N-1,M-1}
    \le
    \tilde{s}_f^{N-1,M-1}
    +
    \sum_{j=0}^{M-2}B_{\max}^j .
\]

We let $\mathcal{C}_{\mathrm{last}}$ denote the time from the start of the iteration until pipeline stage $N-1$ completes the backward computation of microbatch $M-1$. Moreover let $\mathcal{G}=\mathcal{C}-\mathcal{C}_\text{last}$ be the cooldown bubble. Since $\tilde{e}_f^{N-1,M-1}=\mathcal{F}$, we obtain
\[
    C_{\mathrm{last}}
    \le
    \mathcal{F}
    +
    \sum_{j=0}^{M-2}B_{\max}^j
    +
    B_{\mathrm{last}}^{M-1}.
\]

We next bound the remaining backward cooldown $\mathcal{G}$. The backward direction is not exactly symmetric to Lemma~\ref{lemma:forward-delay}, because the induction base changes. For forward execution, the base stage is stage $0$, where delayed forward starts can be charged directly to backward work inserted before them. For backward execution, the base stage is stage $N-1$. In the backward-only schedule, all backward work is initially ready at stage $N-1$. However, backward work at stage $N-1$ can be delayed by forward-side progress needed to make backward tasks ready, including possible idle gaps after a backward task while waiting for the next forward task.

We capture this base-case difference using an augmented forward-side delay term $\hat F_i^k$. For stages $i<N-1$, let $\hat F_i^k=F_i^k$. For the last stage, we let $\hat F_{N-1}^k=e_f^{N-1,k} - e_b^{N-1,k-1}$ if $k>0$ which is no less than $e_f^{N-1,k}-s_f^{N-1,k}=F_{N-1}^k$, and $\hat{F}_{N-1}^0=F_{N-1}^0$. Further define
\[
    \hat F_{\max}^k=\max_i \hat F_i^k,\qquad
    \hat F_{\mathrm{last}}^k=\hat F_{N-1}^k .
\]

\begin{lemma}
\label{lemma:backward-delay}
Set the time origin to the moment when stage $N-1$ starts the backward computation of microbatch $0$, so $s_b^{N-1,0}=0$. Let $s_b^{i,j}$ be the start time of the backward task of microbatch $j$ at stage $i$ under RRFP, and let $\tilde{s}_b^{i,j}$ be the corresponding start time in the backward-only schedule, with the same time origin $\tilde{s}_b^{N-1,0}=0$. Let $J_{i,j}$ denote the set of forward microbatches $k\ge 1$ whose forward task is incurred on stage $i$ before backward task $(i,j)$ starts under RRFP. Then $J_{N-1,0}=\emptyset$, and
\[
    s_b^{i,j}
    \le
    \tilde{s}_b^{i,j}
    +
    \sum_{k\in J_{i,j}}\hat{F}_{\max}^k .
\]
\end{lemma}

\begin{proof}
The proof follows similarly to that of Lemma~\ref{lemma:forward-delay}. Define
\[
D_b^{i,j}=s_b^{i,j}-\tilde{s}_b^{i,j}.
\]
We prove
\[
D_b^{i,j}\le\sum_{k\in J_{i,j}}\hat{F}_{\max}^k
\]
by induction on $(N-1-i)+j$.

In the backward-only schedule,
\[
\tilde{s}_b^{i,j}=
\begin{cases}
0, & i=N-1,j=0,\\
\tilde{e}_b^{i+1,0}, & i<N-1,j=0,\\
\tilde{e}_b^{N-1,j-1}, & i=N-1,j>0,\\
\max\{\tilde{e}_b^{i+1,j},\tilde{e}_b^{i,j-1}\}, & i<N-1,j>0,
\end{cases}
\,\,
\tilde{e}_b^{i,j}=\tilde{s}_b^{i,j}+B_i^j .
\]

We first consider the boundary stage $i=N-1$. Since stage $N-1$ has no downstream backward dependency, the delay of backward task $(N-1,j)$ under RRFP comes only from the augmented local forward work executed before it on stage $N-1$. By the definition of $J_{N-1,j}$, this delay is exactly
\[
    \sum_{k\in J_{N-1,j}}\hat{F}_{\max}^k .
\]
Thus the induction hypothesis holds for $i=N-1$.

Next consider the boundary microbatch $j=0$ and $i<N-1$. In the backward-only schedule, task $(i,0)$ can start once backward task $(i+1,0)$ finishes. Under RRFP, after $e_b^{i+1,0}$, if stage $i$ is free, task $(i,0)$ is dispatched immediately. Otherwise, by the BF hint, task $(i,0)$ can be delayed by at most one forward microbatch at stage $i$. Let this additional delay be $\Delta_{i,0}$. Then
\[
s_b^{i,0}\le e_b^{i+1,0}+\Delta_{i,0}.
\]
Subtracting $\tilde{s}_b^{i,0}=\tilde{e}_b^{i+1,0}$ gives
\[
D_b^{i,0}\le D_b^{i+1,0}+\Delta_{i,0}.
\]
By the induction hypothesis,
\[
D_b^{i+1,0}\le\sum_{k\in J_{i+1,0}}\hat{F}_{\max}^k.
\]
Moreover, the microbatch contributing $\Delta_{i,0}$ is not included in $J_{i+1,0}$, but both $J_{i+1,0}$ and this block are contained in $J_{i,0}$. Therefore,
\[
D_b^{i,0}
\le
\sum_{k\in J_{i+1,0}}\hat{F}_{\max}^k+\Delta_{i,0}
\le
\sum_{k\in J_{i,0}}\hat{F}_{\max}^k .
\]
This proves the induction hypothesis for all $j=0$.

Finally, consider the general case $i<N-1$ and $j>0$. Again, we have
\[
s_b^{i,j}\le\max\{e_b^{i+1,j},\,e_b^{i,j-1}\}+\Delta_{i,j},
\]
where $\Delta_{i,j}$ is the additional delay caused by at most one forward microbatch selected before task $(i,j)$ under the BF hint. Subtracting
\[
\tilde{s}_b^{i,j}=\max\{\tilde{e}_b^{i+1,j},\,\tilde{e}_b^{i,j-1}\}
\]
gives
\[
D_b^{i,j}\le\max\{D_b^{i+1,j},\,D_b^{i,j-1}\}+\Delta_{i,j}.
\]
By the induction hypothesis,
\[
D_b^{i+1,j}\le\sum_{k\in J_{i+1,j}}\hat{F}_{\max}^k,
\qquad
D_b^{i,j-1}\le\sum_{k\in J_{i,j-1}}\hat{F}_{\max}^k .
\]
Moreover, the forward microbatch contributing $\Delta_{i,j}$ is not contained in either $J_{i+1,j}$ or $J_{i,j-1}$, while $J_{i+1,j}$, $J_{i,j-1}$, and this block are all contained in $J_{i,j}$. Therefore,
\[
\begin{aligned}
D_b^{i,j}
&\le
\max\left\{
\sum_{k\in J_{i+1,j}}\hat{F}_{\max}^k,\,
\sum_{k\in J_{i,j-1}}\hat{F}_{\max}^k
\right\}
+\Delta_{i,j} \\
&\le
\sum_{k\in J_{i,j}}\hat{F}_{\max}^k .
\end{aligned}
\]
This completes the induction and proves the lemma.

\end{proof}

Consequently, we arrive at the following lemma.

\begin{lemma}
\label{lemma:backward-cooldown}
The remaining backward cooldown $\mathcal{G}$ is bounded by
\[
    \mathcal{B}
    +
    \sum_{k=1}^{M-1}
    \bigl(F_{\max}^k-F_{\mathrm{last}}^k\bigr)
    -
    \sum_{j=0}^{M-1}B_{\mathrm{last}}^j .
\]
\end{lemma}

\begin{proof}
Set the time origin to the moment when stage $N-1$ starts the backward computation of microbatch $0$. The remaining backward cooldown is determined by the completion time of the last backward task at stage $0$, namely $e_b^{0,M-1}$, after subtracting the portion already accounted for on the stage $N-1$ timeline.

By Lemma~\ref{lemma:backward-delay},
\[
s_b^{0,M-1}-\tilde{s}_b^{0,M-1}
\le
\sum_{k\in J_{0,M-1}}\hat F_{\max}^k
\le
\sum_{k=1}^{M-1}\hat F_{\max}^k .
\]
Since both executions use the same local backward time for task $(0,M-1)$, we have
\[
e_b^{0,M-1}-\tilde e_b^{0,M-1}
=
s_b^{0,M-1}-\tilde{s}_b^{0,M-1}.
\]
Moreover, in the backward-only schedule,
\[
\tilde e_b^{0,M-1}=\mathcal{B}.
\]
Therefore,
\[
e_b^{0,M-1}
\le
\mathcal{B}
+
\sum_{k=1}^{M-1}\hat F_{\max}^k .
\]

The portion already accounted for on the stage $N-1$ timeline consists of all last-stage backward work and the augmented last-stage forward-side terms:
\[
    \sum_{j=0}^{M-1}B_{\mathrm{last}}^j
    +
    \sum_{k=1}^{M-1}\hat F_{\mathrm{last}}^k .
\]
Thus,
\[
\begin{aligned}
\mathcal{G}
&\le
\mathcal{B}
+
\sum_{k=1}^{M-1}\hat F_{\max}^k
-
\sum_{j=0}^{M-1}B_{\mathrm{last}}^j
-
\sum_{k=1}^{M-1}\hat F_{\mathrm{last}}^k \\
&=
\mathcal{B}
+
\sum_{k=1}^{M-1}
\bigl(\hat F_{\max}^k-\hat F_{\mathrm{last}}^k\bigr)
-
\sum_{j=0}^{M-1}B_{\mathrm{last}}^j .
\end{aligned}
\]
Since $\hat F^k$ differs from $F^k$ only by increasing the last-stage term,
\[
    \hat F_{\max}^k-\hat F_{\mathrm{last}}^k
    \le
    F_{\max}^k-F_{\mathrm{last}}^k .
\]
Substituting this inequality gives
\[
    \mathcal{G}
    \le
    \mathcal{B}
    +
    \sum_{k=1}^{M-1}
    \bigl(F_{\max}^k-F_{\mathrm{last}}^k\bigr)
    -
    \sum_{j=0}^{M-1}B_{\mathrm{last}}^j ,
\]
which proves the lemma.
\end{proof}

We now restate and prove Theorem~6.1.

\paragraph{Theorem 6.1.}
The iteration makespan $\mathcal{C}$ of RRFP under the BF hint arbitration rule satisfies
\[
    \mathcal{C}\le
    \mathcal{F}+\mathcal{B}
    +\sum_{j=1}^{M-1}
        \bigl(F_{\max}^j-F_{\mathrm{last}}^j\bigr)
    +\sum_{j=0}^{M-2}
        \bigl(B_{\max}^j-B_{\mathrm{last}}^j\bigr).
\]

\begin{proof}
The makespan can be decomposed into the time until the last-stage backward cooldown begins and the remaining cooldown:
\[
    \mathcal{C}
    =
    C_{\mathrm{last}}
    +
    \mathcal{G}.
\]
Using the bounds above,
\[
\begin{aligned}
\mathcal{C}
\le\;&
\mathcal{F}+\mathcal{B}
+\sum_{j=1}^{M-1}\bigl(F_{\max}^j-F_{\mathrm{last}}^j\bigr)
+\sum_{j=0}^{M-2}B_{\max}^j
-\sum_{j=0}^{M-2}B_{\mathrm{last}}^j \\
=\;&
\mathcal{F}+\mathcal{B}
+\sum_{j=1}^{M-1}\bigl(F_{\max}^j-F_{\mathrm{last}}^j\bigr)
+\sum_{j=0}^{M-2}\bigl(B_{\max}^j-B_{\mathrm{last}}^j\bigr).
\end{aligned}
\]
This proves the theorem.
\end{proof}

\subsection{Proof of Corollary~6.2}
\label{app:proof-cor-6-2}

Define $T_{\max}$ as
\[
    T_{\max}=\max_{0\le j\le M-1}\bigl(F_{\max}^j+B_{\max}^j\bigr),
\]

We first record a simple fill/drain bound for the forward-only and backward-only makespans.

\begin{lemma}
\label{lemma:fb-fill-drain}
The single-direction makespans satisfy
\[
    \mathcal{F}
    \le
    \sum_{j=0}^{M-1}F_{\max}^j+O(NT_{\max}),
    \qquad
    \mathcal{B}
    \le
    \sum_{j=0}^{M-1}B_{\max}^j+O(NT_{\max}).
\]
\end{lemma}

\begin{proof}
We prove the forward case. The backward case is analogous. Using the notations $\tilde{s}_f^{i,j}$ and $\tilde{e}_f^{i,j}$ from above, we have
\[
\tilde{s}_f^{i,j}=\max\{\tilde{e}_f^{i-1,j},\tilde{e}_f^{i,j-1}\},
\]
where we ignore non-existing terms at the boundary stages and microbatches. We claim that
\[
\tilde{s}_f^{i,j}\le\sum_{t=0}^{j-1}F_{\max}^t+iT_{\max}.
\]
We prove this by induction on $i+j$. The base case $(i,j)=(0,0)$ is immediate. For the induction step, by the induction hypothesis,
\[
\begin{aligned}
\tilde{e}_f^{i,j-1}
&=
\tilde{s}_f^{i,j-1}+F_i^{j-1}\\
&\le
\sum_{t=0}^{j-2}F_{\max}^t+iT_{\max}+F_{\max}^{j-1}\\
&=
\sum_{t=0}^{j-1}F_{\max}^t+iT_{\max},
\end{aligned}
\]
and
\[
\begin{aligned}
\tilde{e}_f^{i-1,j}
&=
\tilde{s}_f^{i-1,j}+F_{i-1}^{j}\\
&\le
\sum_{t=0}^{j-1}F_{\max}^t+(i-1)T_{\max}+F_{\max}^{j}
\le
\sum_{t=0}^{j-1}F_{\max}^t+iT_{\max},
\end{aligned}
\]
where the last inequality uses $F_{\max}^j\le T_{\max}$. Therefore,
\[
\tilde{s}_f^{i,j}\le\sum_{t=0}^{j-1}F_{\max}^t+iT_{\max}.
\]
Thus,
\[
\begin{aligned}
\tilde{e}_f^{N-1,M-1}
&\le
\sum_{t=0}^{M-2}F_{\max}^t+(N-1)T_{\max}+F_{\max}^{M-1}\\
&=
\sum_{t=0}^{M-1}F_{\max}^t+O(NT_{\max}).
\end{aligned}
\]
Hence
\[
\mathcal{F}
\le
\sum_{j=0}^{M-1}F_{\max}^j+O(NT_{\max}).
\]
Applying the same argument to backward gives
\[
\mathcal{B}
\le
\sum_{j=0}^{M-1}B_{\max}^j+O(NT_{\max}).
\]
which proves the lemma.
\end{proof}

\begin{figure*}[t]
    \centering
    \includegraphics[width=\linewidth]{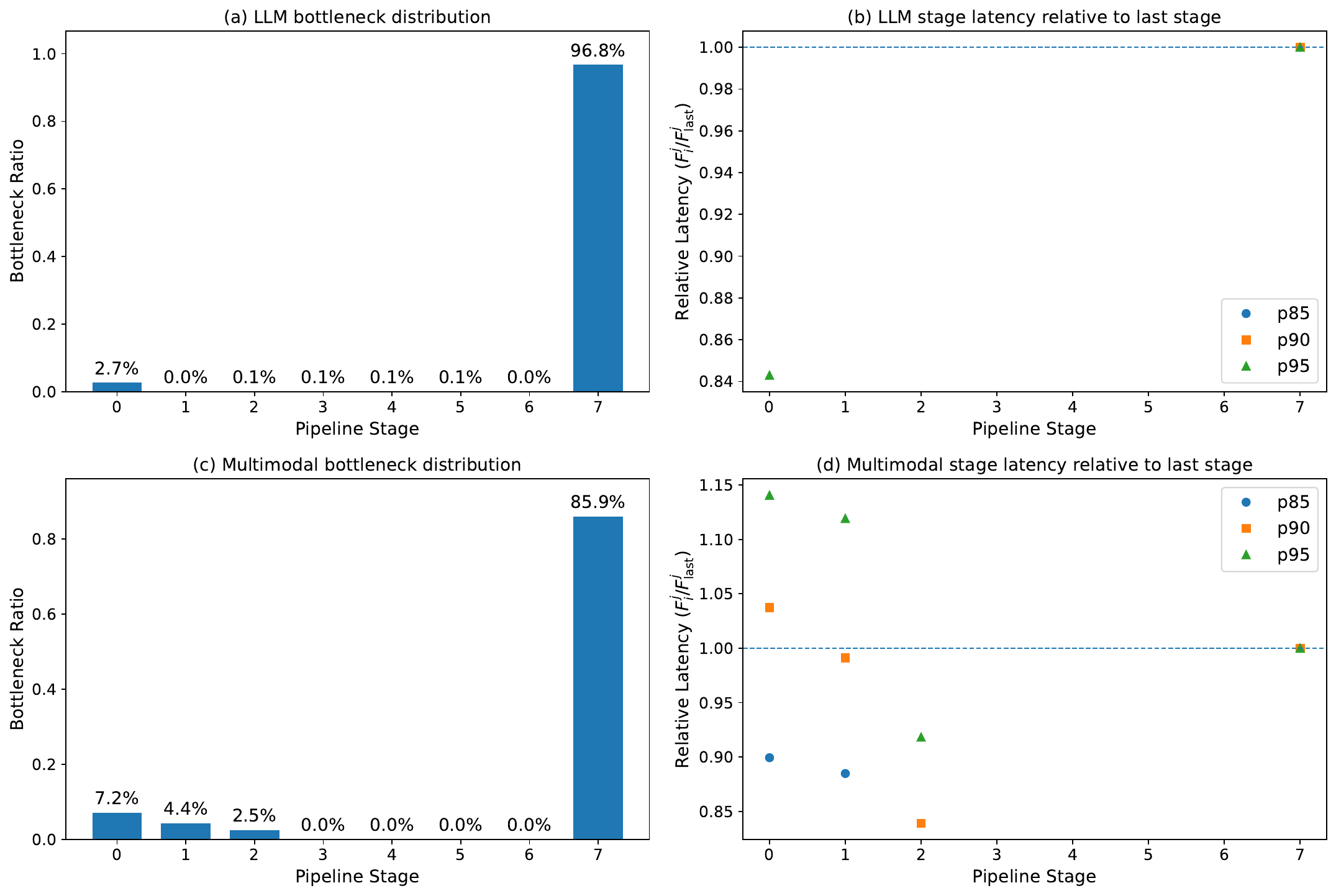}
    \caption{
    \textbf{Detailed bottleneck statistics over 100 iterations.}
    Panels (a) and (c) report the fraction of forward microbatches for which each pipeline stage is the bottleneck. The last stage dominates in both settings, accounting for 96.8\% of bottleneck cases in the LLM workload and 85.9\% in the multimodal workload.
    Panels (b) and (d) report p85, p90, and p95 relative forward latency, normalized by the last-stage latency, i.e., $F_i^j/F_{\mathrm{last}}^j$. The maximum p95 non-last-stage ratio is 0.84$\times$ for LLM and 1.14$\times$ for multimodal, supporting the bounded-deviation assumption used in Corollary~6.2.
    }
    \label{fig:app-bottleneck}
    \Description{Bottleneck.}
\end{figure*}

Let
\[
    L=
    \sum_{j=0}^{M-1}
    \bigl(F_{\mathrm{last}}^j+B_{\mathrm{last}}^j\bigr).
\]
For any valid offline or online schedule, $\mathrm{OPT}\ge L$, because every schedule must execute all forward and backward tasks on the last pipeline stage.

Assume that for each microbatch $j$, the last stage is the bottleneck with probability at least $1-p$. Otherwise, assume that the deviation from the last-stage time is bounded by a constant factor $\rho\ge1$:
\[
    F_{\max}^j\le \rho F_{\mathrm{last}}^j,
    \qquad
    B_{\max}^j\le \rho B_{\mathrm{last}}^j .
\]
and we further assume that the computation time of each microbatch is bounded within $[m_l,m_h]$. Hence,
\[
T_{\max}\le 2m_h \le \frac{L}{M}\frac{m_h}{m_l}
= O\left(\frac{L}{M}\right),
\]
where $m_l$ and $m_h$ are constants.

\paragraph{Corollary~6.2.} Under the assumptions above,
\[
    \mathbb{E}\!\left[\frac{\mathcal{C}}{\mathrm{OPT}}\right]
    \le
    1+2p(\rho-1)+O\!\left(\frac{N}{M}\right).
\]

\begin{proof}
Define the stage-imbalance terms
\[
    \Delta_f^j=F_{\max}^j-F_{\mathrm{last}}^j,
    \qquad
    \Delta_b^j=B_{\max}^j-B_{\mathrm{last}}^j .
\]
By Lemma~\ref{lemma:fb-fill-drain},
\[
    \mathcal{F}+\mathcal{B}
    \le
    \sum_{j=0}^{M-1}\bigl(F_{\max}^j+B_{\max}^j\bigr)
    +O(NT_{\max}).
\]
Using
\[
    F_{\max}^j=F_{\mathrm{last}}^j+\Delta_f^j,
    \qquad
    B_{\max}^j=B_{\mathrm{last}}^j+\Delta_b^j,
\]
we obtain
\[
    \mathcal{F}+\mathcal{B}
    \le
    L+
    \sum_{j=0}^{M-1}(\Delta_f^j+\Delta_b^j)
    +O(NT_{\max}).
\]
Substituting this into Theorem~6.1 gives
\[
\begin{aligned}
    \mathcal{C}
    \le\;&
    L+
    \sum_{j=0}^{M-1}(\Delta_f^j+\Delta_b^j)
    +
    \sum_{j=1}^{M-1}\Delta_f^j
    +
    \sum_{j=0}^{M-2}\Delta_b^j
    +O(NT_{\max}) \\
    \le\;&
    L+
    2\sum_{j=0}^{M-1}(\Delta_f^j+\Delta_b^j)
    +O(NT_{\max}).
\end{aligned}
\]

Under the last-stage dominance assumption, for each microbatch $j$,
\[
    \mathbb{E}[\Delta_f^j]
    \le
    p(\rho-1)F_{\mathrm{last}}^j,
    \qquad
    \mathbb{E}[\Delta_b^j]
    \le
    p(\rho-1)B_{\mathrm{last}}^j .
\]
Therefore,
\[
    \mathbb{E}
    \!\left[
        \sum_{j=0}^{M-1}
        (\Delta_f^j+\Delta_b^j)
    \right]
    \le
    p(\rho-1)L .
\]
Taking expectation and using $\mathrm{OPT}\ge L$ yields
\[
    \mathbb{E}\!\left[\frac{\mathcal{C}}{\mathrm{OPT}}\right]
    \le
    1+2p(\rho-1)
    +
    O\!\left(\frac{NT_{\max}}{L}\right).
\]
Since $T_{\max}=O(L/M)$, the final term is $O(N/M)$. Hence,
\[
    \mathbb{E}\!\left[\frac{\mathcal{C}}{\mathrm{OPT}}\right]
    \le
    1+2p(\rho-1)+O\!\left(\frac{N}{M}\right),
\]
as claimed.
\end{proof}

\subsection{Detailed Bottleneck Statistics}
\label{app:bottleneck-stats}

Figure~\ref{fig:app-bottleneck} reports the detailed statistics behind the last-stage-dominance assumption used in Corollary~6.2. We measure forward computation times over 100 iterations and identify, for each forward microbatch, which pipeline stage has the largest forward computation time. The bottleneck distributions show that the last stage dominates most forward microbatches: 96.8\% in the LLM workload and 85.9\% in the multimodal workload.

The relative-latency panels further show that deviations from last-stage dominance are bounded in scale. For each stage $i$, we report the p85, p90, and p95 of the relative latency $F_i^j/F_{\mathrm{last}}^j$ across microbatches. In the LLM workload, the maximum p95 relative latency among non-last stages is 0.84$\times$, indicating that non-last stages are consistently below the last-stage latency. In the multimodal workload, some non-last stages can exceed the last-stage latency, but the maximum p95 relative latency is still only 1.14$\times$. These measurements support the bottleneck assumption used in Corollary~6.2: the last stage is usually the bottleneck, and deviations from last-stage dominance remain within a small constant factor.

\section{Detailed Buffer-Size Policy Analysis and Deadlock-Free Guarantee}
\label{app:buffer}
\subsection{Buffer Implementation Details}
\label{sec:2.1}
At each pipeline-parallel stage, RRFP maintains four buffers ($N$ denotes PP size):
\begin{itemize}
\item Forward-ready buffer: stores forward inputs received from the previous pipeline stage, except at PP stage $0$ in non-interleaved schedules.
\item Forward-finished buffer: stores forward outputs waiting to be sent to the next pipeline stage.
\item Backward-ready buffer: stores backward inputs received from the next pipeline stage.
\item Backward-finished buffer: stores backward outputs waiting to be sent to the previous pipeline stage, except at PP stage $0$ in non-interleaved schedules.
\end{itemize}

RRFP uses the following release policy. First, tensors in the forward-finished and backward-finished buffers are released immediately after the corresponding send completes. Second, each tensor in the backward-ready buffer corresponds to a retained tensor in the forward-ready buffer for the same microbatch. The forward-ready tensor is retained until the corresponding backward computation completes. At that point, both the retained forward input and the backward input are released. Therefore, releasing a tensor from the forward-ready buffer requires the corresponding backward input to have arrived and completed computation.

For PP stage $N-1$, RRFP still sets up the forward-finished and backward-ready buffers, even in non-interleaved schedules. But the runtime moves the forward output directly into the backward-ready buffer locally.

\subsection{Buffer Backpressure Policy}
In RRFP, for PP stage $i$, let $n_f^i$ denote the number of forward computations executed by stage $i$ so far in the current iteration, and let $n_b^i$ denote the number of backward computations executed by stage $i$ so far in the current iteration. The $i$th stage computes
\[
    \mathcal{D}_i=n_f^i - n_b^i .
\]
before each arbitration. If this value is greater than or equal to the configured limit, the stage enters backpressure mode.

For non-interleaved schedules, each microbatch passes through each pipeline stage only once in the forward direction and once in the backward direction. Therefore, once $\mathcal{D}_i$ reaches the configured limit, the arbitration layer temporarily disables forward dispatch and follows a backward-only drain order. It executes a ready backward task whenever one is available. If no backward task is ready, the stage waits instead of skipping to ready forward tasks. This prevents the stage from producing more forward outputs while allowing backward computations to release the corresponding retained forward inputs. Hence $\mathcal{D}_i$ cannot increase in this mode and decreases whenever a backward task executes. The stage remains in this mode until $\mathcal{D}_i$ drops below the limit.

For interleaved schedules, the situation is more subtle because a microbatch traverses the physical pipeline multiple times, once for each model chunk. For example, with four physical pipeline stages and multiple virtual pipeline stages per stage, the forward path of a microbatch may traverse the stages as $0,1,2,3$ for one model chunk, then again as $0,1,2,3$ for the next model chunk, and so on. The backward path traverses these chunks in the reverse order. In this case, simply disabling all forward tasks can be unsafe for progress: a microbatch may still need to finish later model chunks in the forward pass before any of its backward tasks can become ready, and waiting only for backward tasks may therefore prevent the computations that would eventually enable buffer release.

Thus, under backpressure in interleaved schedules, RRFP switches to a deterministic order that completes microbatches one by one. The stage scans microbatches from index $0$ to $M-1$, where $M$ is the total number of microbatches. For each unfinished microbatch, the stage follows the microbatch's fixed local completion order. If a microbatch needs to traverse local model chunks $0,\ldots,C-1$, where $C$ is the number of model chunks, this order is
$F_0,F_1,\ldots,F_{C-1},B_{C-1},\ldots,B_1,B_0$:
the forward computations proceed from earlier chunks to later chunks, and the backward computations proceed in the reverse chunk order. The stage checks where the current microbatch is in this order and attempts to execute the next required task. If that task is ready, the stage executes it. If it is not ready, the stage waits rather than skipping to a later task. For example, if microbatch $i$ has completed $F_0,F_1,F_2$ but $F_3$ has not yet arrived, the stage waits for $F_3$. After executing $F_3$, it waits for $B_3$, then $B_2$, and so on until the microbatch is fully drained. In effect, the backpressure policy stops opportunistic dispatch and concentrates progress on completing earlier microbatches, so that their backward computations can retire the corresponding activations and reduce buffer occupancy. In this case, each focused microbatch can increase $\mathcal{D}_i$ by at most $C$ additional forward computations before its backward computations decrease it exactly by $C$. Therefore, $\mathcal{D}_i$ can exceed the limit by at most $C$ under backpressure. When $\mathcal{D}_i$ drops below the limit, the stage exits backpressure mode and returns to the normal readiness-driven arbitration policy.

\subsection{Effectiveness of Backpressure}
We first establish the following theorem.

\begin{theorem}
Each individual buffer has size at most $\mathcal{D}_0$.
\end{theorem}

Note that for interleaved schedules, the same microbatch may occupy multiple slots, even within the same buffer, corresponding to computation inputs or outputs at different model chunks.

\begin{proof}
For both interleaved and non-interleaved schedules, any buffer entry can appear only after stage $0$ has executed the corresponding forward computation (in the same model chunk), and it leaves the system once stage $0$ has executed the corresponding backward computation (in the same model chunk). Therefore, every existing buffer entry corresponds to one forward computation at stage $0$ whose matching backward computation has not yet been executed, so the total number of such entries is at most $\mathcal{D}_0$.
\end{proof}

\begin{corollary}
With a configured limit $M$, every buffer has size at most $M$ for non-interleaved schedules and at most $M+C$ for interleaved schedules.
\end{corollary}

\begin{proof}
Under backpressure, $\mathcal{D}_0$ cannot increase beyond $M$ in non-interleaved schedules. In interleaved schedules, $\mathcal{D}_0$ may increase by at most $C$ additional forward computations after entering backpressure mode, but the corresponding backward computations then decrease it by exactly $C$. Hence $\mathcal{D}_0$ is bounded by $M$ and $M+C$, respectively, and the claim follows from the theorem above.
\end{proof}

\subsection{Deadlock-Free Guarantee}
We finally prove that training with RRFP under the backpressure policy is deadlock-free.

\begin{theorem}
Under reliable computation and communication, training with RRFP under the backpressure policy is deadlock-free if the configured limit is positive.
\end{theorem}

\begin{proof}
We first consider non-interleaved schedules. A deadlock can occur only if no stage has executable work. At any time $t$, we have
\[
    n_b^i \le n_b^{i+1}.
\]
Moreover, if $n_b^i < n_b^{i+1}$, then stage $i$ has a backward-ready task and therefore cannot be stuck due to reliable communication. Hence, a deadlock can occur only if
\[
    n_b^0=n_b^1=\cdots=n_b^{N-1}.
\]

Now consider the last stage $N-1$. Under the BF hint, this stage follows an exact 1F1B pattern. Therefore,
\[
    n_f^{N-1}-n_b^{N-1}\in\{0,1\}.
\]
Even when the buffer limit is $1$, backpressure can only force the last stage to execute a backward task, which is exactly the usual 1F1B behavior because stage $N-1$ executes the backward task immediately after the corresponding forward task anyway. Thus, for stage $N-1$ to be stuck, we must have
\[
    n_f^{N-1}=n_b^{N-1}.
\]
Otherwise, stage $N-1$ has backward-ready work.

Next, at any time $t$ we have
\[
    n_f^i \ge n_f^{i+1}
\]
for every stage $i$. Suppose not all forward counters are equal. Then there exists some stage $i<N-1$ such that
\[
    n_f^i > n_f^{i+1}=n_f^{i+2}=\cdots=n_f^{N-1}.
\]
This means that stage $i+1$ has a forward-ready task that has already been output by stage $i$ but has not yet been executed at stage $i+1$. Moreover, from the previous argument and the equality of all backward counters, we have
\[
    n_f^{i+1}=n_f^{N-1}=n_b^{N-1}=n_b^{i+1}.
\]
Therefore, stage $i+1$ is not in backpressure mode and can execute this forward-ready task, contradicting the assumption of deadlock.

It remains to consider the case
\[
    n_f^0=n_f^1=\cdots=n_f^{N-1}.
\]
Together with $n_f^{N-1}=n_b^{N-1}$ and the equality of all backward counters, this implies that the current set of admitted microbatches has fully completed both forward and backward execution. In this case, stage $0$ can admit and initiate subsequent microbatches, so the system is not deadlocked.

Therefore, no deadlock can occur in the non-interleaved case.

For the interleaved case, first observe that
\[
    n_f^i-n_b^i=\mathcal{D}_i \ge \mathcal{D}_{i+1}=n_f^{i+1}-n_b^{i+1}.
\]
Therefore, if stage $i+1$ is in backpressure mode, then stage $i$ must also be in backpressure mode.

\paragraph{Case 1: No stage is in backpressure mode.}
By the same argument as in the non-interleaved case, if the system has no executable work, then we must have
\[
    n_f^0=\cdots=n_f^{N-1}
    \qquad\text{and}\qquad
    n_b^0=\cdots=n_b^{N-1}.
\]
Otherwise, some stage would have either forward-ready or backward-ready work.

Now consider stage $N-1$. If
\[
    n_f^{N-1} > n_b^{N-1},
\]
then stage $N-1$ has at least one forward microbatch whose corresponding backward computation has not yet started. In an interleaved schedule, this means that either the microbatch is eligible for backward, or has not yet finished traversing all local model chunks. In the first case, stage $N-1$ has work to do which is impossible. In the second case, such a microbatch must be forwarded from stage $N-1$ back to stage $0$ for the next model chunk, thereby creating an additional forward-ready task at stage $0$. This contradicts the assumption that no executable work exists. Hence, we must have $n_f^{N-1}=n_b^{N-1}$, and together with the equalities above, all admitted work has completed. Stage $0$ can therefore admit subsequent microbatches, so the system is not deadlocked.

\paragraph{Case 2: All stages are in backpressure mode.}
In this case, a deadlock could only occur if each stage is waiting for its next target microbatch. Suppose stage $0$ is waiting for a forward microbatch $m_{0,f}^{j,k}$, where $j$ denotes the microbatch index and $k$ denotes the model-chunk index. We must have $k\ge 1$. Otherwise, stage $0$ could start computing the first chunk directly. Therefore, stage $0$ must have already finished computing $m_{0,f}^{j,k-1}$ and forwarded it to the subsequent stages.

Under the backpressure policy, as this microbatch becomes ready at each stage $i$, the corresponding task $m_{i,f}^{j,k-1}$ is the target microbatch and must be computed immediately. Hence, under reliable computation and communication, it can progress through all stages until it reaches stage $0$ again as the next model chunk $m_{0,f}^{j,k}$. Consequently, the microbatch that stage $0$ is waiting for will eventually become ready, contradicting the assumption of deadlock.

The case where stage $0$ is waiting for a backward microbatch is analogous: if it waits for microbatch $m_{0,b}^{j,k}$, then it must have finished the previous backward microbatch $m_{0,b}^{j,k+1}$, or the corresponding forward microbatch $m_{0,f}^{j,k}$ when $k$ is the last model chunk. Since this microbatch is the target microbatch for the other stages (all other stages are waiting for it), it is computed immediately. Therefore, the backward microbatch that stage $0$ is waiting for will eventually become ready, again contradicting deadlock.

\paragraph{Case 3: A prefix of stages is in backpressure mode.}
It remains to consider the mixed case where stages $0,\ldots,k$ are in backpressure mode while stages $k+1,\ldots,N-1$ are not, for some $0\le k<N-1$. Since
\[
    \mathcal{D}_i \ge \mathcal{D}_{i+1},
\]
the set of stages in backpressure mode must form such a prefix. Assume, for contradiction, that the system is deadlocked. Then the suffix stages $k+1,\ldots,N-1$ cannot have any ready work. Otherwise, one of them could execute. Therefore, the suffix has already consumed all work currently available from the prefix. In particular, the microbatch that the prefix is waiting for is not stuck inside the suffix, otherwise, the suffix will have executable work, contradicting deadlock.

Thus, the only possible waiting cycle lies within the prefix $0,\ldots,k$, where all stages are in backpressure mode. We can view stage $k$ as the boundary stage of this prefix: forward work leaving stage $k$ returns to stage $0$ for the next chunk or returns to stage $k$ for last-model-chunk backward, while backward work leaving stage $0$ returns to stage $k$ for the previous chunk. Hence, the suffix does not introduce an additional blocking point. Therefore, the same argument as in Case~2 applies to the prefix: the target microbatch that the prefix is waiting for cannot be permanently blocked and must eventually return to the corresponding boundary stage, contradicting deadlock.

Combining the three cases, no deadlock can occur in the interleaved case.

\end{proof}

\section{Collective-Order Consistency and Progress}
\label{app:tp-progress}

This section gives the detailed argument for RRFP's tensor-parallel coordination protocol. The argument concerns collective-order consistency and progress within each tensor-parallel group. Training-loss validation is reported separately in Section~\ref{app:loss-curves}.

\subsection{Coordination Invariant}

Consider one tensor-parallel group. Before any computation step that may invoke tensor-parallel collectives, each rank selects a local candidate from its ready set using the deterministic ready-set arbitration rule. The ranks then exchange the selected microbatch identifiers using a scalar metadata all-gather. If all ranks select the same identifier, the group enters the computation step and therefore invokes collectives for the same tensor in the same order. If the identifiers differ, no rank enters the collective-relevant computation step. The step is deferred, and the local arbitration state is not advanced. The next attempt restarts from the beginning of the same hint order. For local 1F1B hints, this means retrying from backward-forward probing rather than from a later phase.

We maintain the following invariant: after every successful collective-relevant computation step, all ranks in the tensor-parallel group have executed the same sequence of collective-relevant microbatches. This invariant holds initially because no collective-relevant task has been executed. Suppose it holds after some successful step. Before the next step, ranks may observe different ready sets because activations or gradients can arrive at slightly different times. The coordination protocol prevents this temporary disagreement from becoming inconsistent collective execution: the group defers computation, avoids entering mismatched collectives, and retries arbitration when later message arrivals change the ready set. If ranks select the same candidate, they all execute the same collective-relevant task, so the executed sequence remains identical across the group. By induction, every successful collective-relevant step preserves the invariant, and tensor-parallel collectives are invoked in a consistent order.

\subsection{Progress Under Deferred Computation}

Deferring a computation step does not stop communication progress. RRFP's send and receive threads run independently of the compute thread, so missing activations or gradients can still arrive while collective-relevant computation is deferred. When such messages arrive, they update the corresponding receive buffers and may add new candidates to the local ready sets.

Under reliable message delivery, a temporary disagreement among tensor-parallel ranks can therefore be resolved by subsequent readiness events. Once missing inputs arrive, ranks that previously lacked a candidate can observe it in their ready set. After a failed coordination attempt, ranks do not consume the candidate or move to a later position in the hint order. They retry arbitration from the same point after the ready set changes. Since arbitration is deterministic for a given ready set and hint order, the tensor-parallel group can expose a common candidate and pass the metadata agreement check. The group then resumes computation and enters the next collective-relevant step consistently. Thus, RRFP avoids mismatched collectives while allowing communication to continue until tensor-parallel ranks agree on a common task.

\section{Training Correctness Validation}
\label{app:loss-curves}

This section provides the detailed loss-curve validation referenced in the main paper. The goal is to check whether RRFP preserves the intended training behavior relative to the corresponding 1F1B baseline when only the runtime execution order is changed.

\paragraph{Validation protocol.}
We evaluate three representative configurations: GPT-Large with TP1/PP8/DP1 and batch size 64, Qwen3-1.7B+ViT-H with TP1/PP8/DP1 and batch size 96, and Qwen3-1.7B+ViT-H with TP2/PP8/DP1 and batch size 96. For each configuration, we run both RRFP and 1F1B with three seeds: 42, 1234, and 2026. Each run uses 500 training iterations, with evaluation every 50 iterations and 10 evaluation iterations per evaluation. We report training and validation loss curves using a logarithmic y-axis, where equal vertical distances correspond to multiplicative changes in loss, to make the rapid early decrease and later convergence behavior visible.

\begin{figure*}[t]
    \centering
    \includegraphics[width=\textwidth]{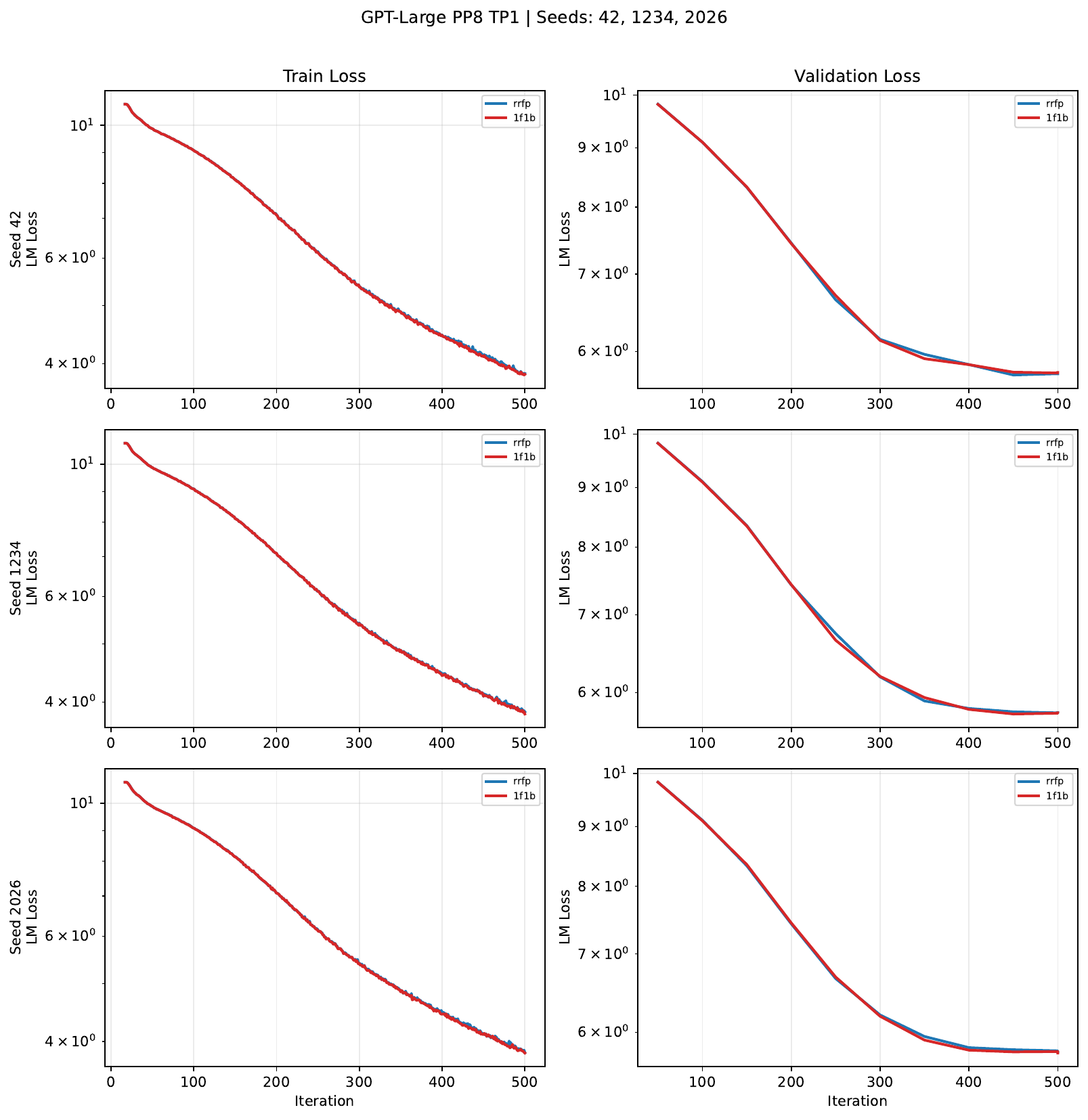}
    \caption{
    \textbf{Training-correctness validation for GPT-Large with TP1/PP8/DP1 and batch size 64.}
    Each row corresponds to one seed, and the two columns report training and validation loss, respectively.
    RRFP and 1F1B show nearly overlapping loss curves across all seeds.
    }
    \Description{Training and validation loss curves for GPT-Large TP1/PP8/DP1 comparing RRFP and 1F1B across three seeds.}
    \label{fig:app-loss-gpt}
\end{figure*}

\begin{figure*}[t]
    \centering
    \includegraphics[width=\textwidth]{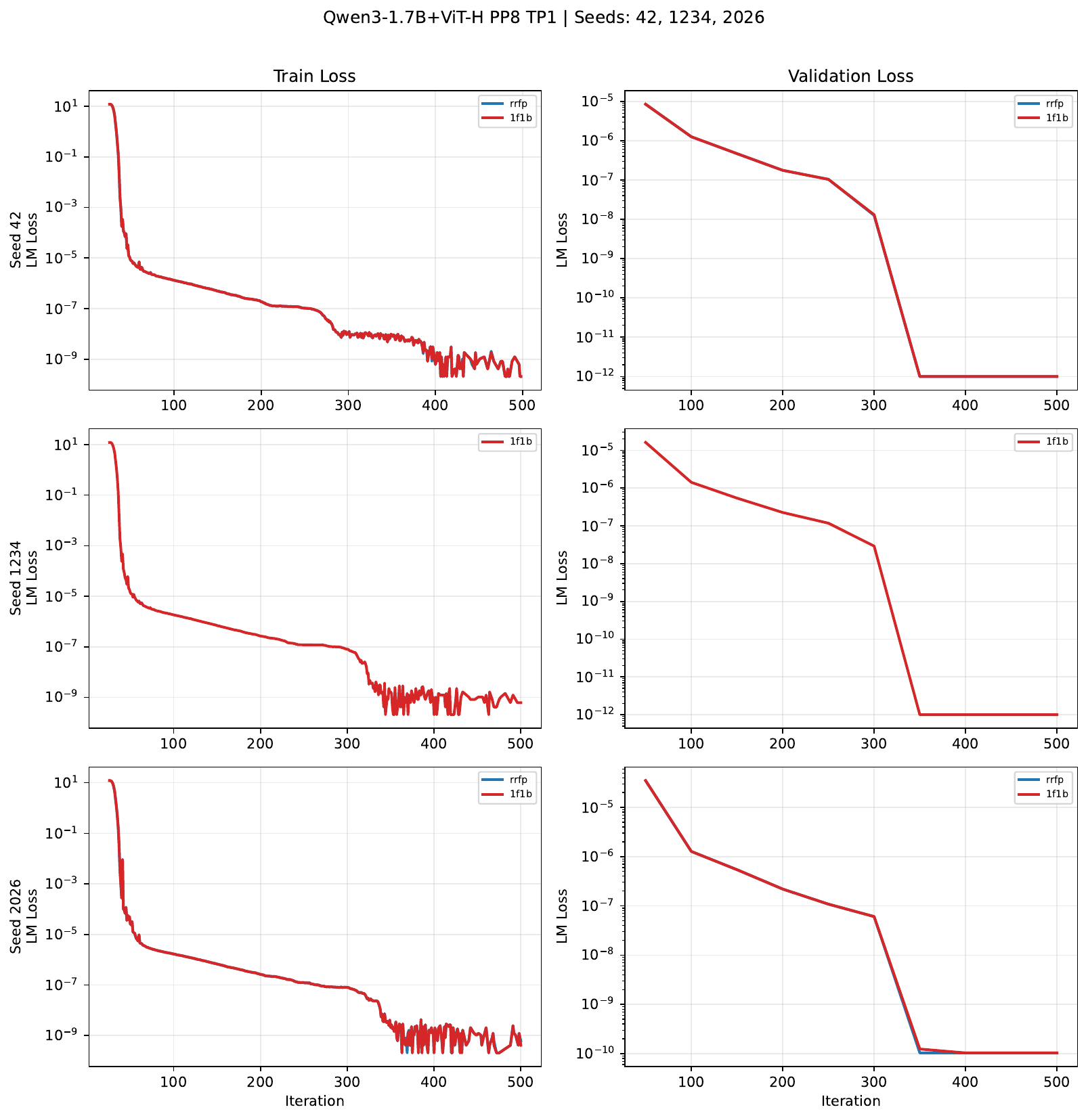}
    \caption{
    \textbf{Training-correctness validation for Qwen3-1.7B+ViT-H with TP1/PP8/DP1 and batch size 96.}
    Each row corresponds to one seed, and the two columns report training and validation loss, respectively.
    RRFP and 1F1B show nearly overlapping loss curves across all seeds.
    }
    \Description{Training and validation loss curves for Qwen3-1.7B+ViT-H TP1/PP8/DP1 comparing RRFP and 1F1B across three seeds.}
    \label{fig:app-loss-qwen-tp1}
\end{figure*}

\begin{figure*}[t]
    \centering
    \includegraphics[width=\textwidth]{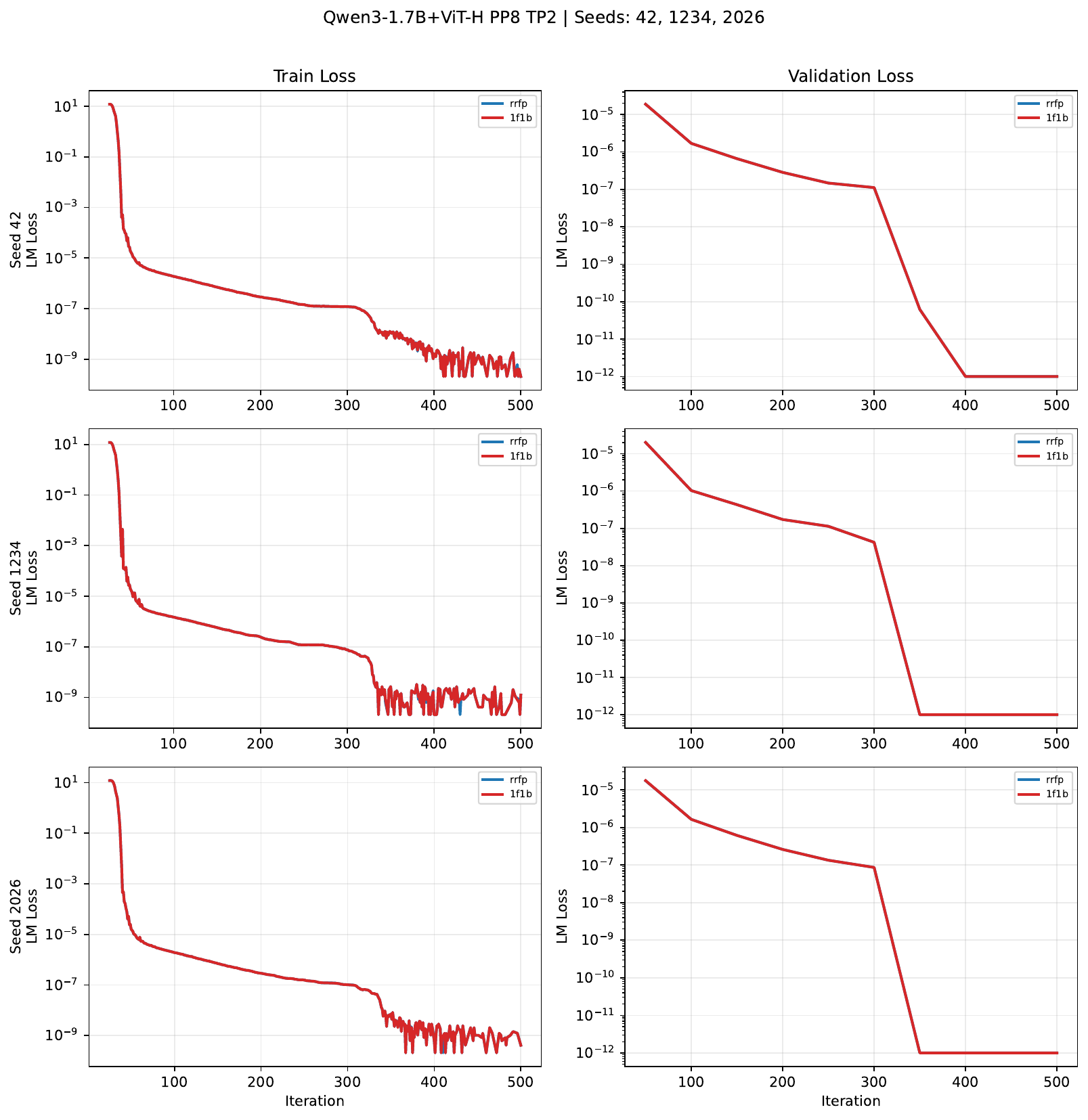}
    \caption{
    \textbf{Training-correctness validation for Qwen3-1.7B+ViT-H with TP2/PP8/DP1 and batch size 96.}
    Each row corresponds to one seed, and the two columns report training and validation loss, respectively.
    RRFP and 1F1B show nearly overlapping loss curves across all seeds.
    }
    \Description{Training and validation loss curves for Qwen3-1.7B+ViT-H TP2/PP8/DP1 comparing RRFP and 1F1B across three seeds.}
    \label{fig:app-loss-qwen-tp2}
\end{figure*}

\paragraph{Results.}
Figures~\ref{fig:app-loss-gpt}--\ref{fig:app-loss-qwen-tp2} show that RRFP follows the same overall convergence trend as 1F1B across all three configurations. For GPT-Large, the training and validation curves nearly overlap throughout training. For Qwen, both TP1 and TP2 settings rapidly reduce loss by several orders of magnitude and reach comparable final loss levels. Across seeds, we observe no RRFP-specific divergence or systematic loss instability. These results support that readiness-driven execution changes the runtime ordering of executable tasks while preserving the intended training semantics.

\end{document}